\documentclass[a4paper,USenglish,cleveref,autoref,thm-restate]{lipics-v2021}
\usepackage{amsmath}
\usepackage{amssymb}
\usepackage{amsthm}
\usepackage{mathtools}
\usepackage{cite}
\usepackage{xspace}
\usepackage{siunitx}
\usepackage{tcolorbox}

\DeclareFontShape{T1}{lmr}{m}{scit}{<->ssub * lmr/m/scsl}{}

\usepackage{complexity}
\usepackage{booktabs}

\theoremstyle{definition}

\newcommand{\ETH}{\textup{ETH}\xspace}

\DeclareMathOperator{\fst}{fst}
\DeclareMathOperator{\snd}{snd}
\DeclareMathOperator{\val}{val}
\DeclareMathOperator{\opt}{opt}
\DeclareMathOperator{\copt}{copt}

\DeclareMathOperator{\first}{first}
\DeclareMathOperator{\last}{last}
\DeclareMathOperator{\cost}{cost}

\newcommand{\dist}{\operatorname{dist}}

\newtheorem{longtheorem}{Theorem}
\newtheorem{longlemma}[longtheorem]{Lemma}
\newtheorem{longredrule}{Reduction Rule}

\newtheorem{longdefinition}[longtheorem]{Definition}
\newtheorem{longobservation}[longtheorem]{Observation}

\usepackage{alphalph}

\crefname{longtheorem}{Theorem}{Theorems}
\Crefname{longtheorem}{Theorem}{Theorems}
\crefname{longlemma}{Lemma}{Lemmas}
\Crefname{longlemma}{Lemma}{Lemmas}
\crefname{longdefinition}{Definition}{Definitions}
\Crefname{longdefinition}{Definition}{Definitions}
\crefname{longobservation}{Observation}{Observations}
\Crefname{longobservation}{Observation}{Observations}
\crefname{longclaim}{Claim}{Claims}
\Crefname{longclaim}{Claim}{Claims}
\crefname{longcorollary}{Corollary}{Corollaries}
\Crefname{longcorollary}{Corollary}{Corollaries}
\crefname{longredrule}{Rule}{Rules}
\Crefname{longredrule}{Rule}{Rules}
\crefrangeformat{longlemma}{Lemmas~#3#1#4--#5#2#6}

\newcommand{\sel}{\mathsf{sel}}

\newcommand{\traces}{\mathsf{traces}}
\newcommand{\data}{\mathsf{data}}
\newcommand{\startsLeft}{\mathsf{startsLeft}}
\newcommand{\stopsRight}{\mathsf{stopsRight}}
\newcommand{\openPred}{\mathsf{openPred}}
\newcommand{\openSucc}{\mathsf{openSucc}}
\newcommand{\forgotten}{\mathsf{forgotten}}

\newcommand{\prob}{\textup{\textsc{TTPVS}}\xspace}
\newcommand{\problong}{\textup{\textsc{Tree-constrained Two-layer Planar Vertex Splitting}}\xspace}
\newcommand{\probStar}{\textup{\textsc{Two-layer Planar Vertex Splitting}}\xspace}
\newcommand{\povs}{\textup{\textsc{Pathwidth-One Vertex Splitting}}\xspace}

\newcommand{\wts}{\textup{\textsc{WTS}}\xspace}
\newcommand{\wtsLong}{\textup{\textsc{Weighted Tuple Sorting}}\xspace}

\newcommand{\problembox}[3]{%
  \begin{tcolorbox}[
    colback=white,
    colframe=black,
    boxrule=\fboxrule,
    boxsep=0.5em,
    left=0pt,
    right=0pt,
    top=0pt,
    bottom=0pt,
    sharp corners,
    before skip=3mm,
    after skip=2mm
  ]
    #1\par\smallskip
    \textbf{Input:} #2\par
    \textbf{Task:} #3
  \end{tcolorbox}%
}

\hideLIPIcs

\title{Two-Layer Drawings with a Tree on Top: \texorpdfstring{\\}{ }
Vertex Splits and Fixed-Parameter Algorithms
}

\titlerunning{Two-Layer Drawings with a Tree on Top}

\author{Alexander Firbas}{Algorithms and Complexity Group, TU Wien, Austria}{afirbas@ac.tuwien.ac.at}{https://orcid.org/0009-0007-2049-2144}{}

\author{Robert Ganian}{Algorithms and Complexity Group, TU Wien, Austria}{rganian@gmail.com}{https://orcid.org/0000-0002-7762-8045}{}

\author{Sylvain Meunier}{Univ Rennes, F-35000 Rennes, France}{sylvain.meunier@ens-rennes.fr}{https://orcid.org/0009-0004-6563-3359}{}

\author{Martin  Nöllenburg}{Algorithms and Complexity Group, TU Wien, Austria}{noellenburg@ac.tuwien.ac.at}{https://orcid.org/0000-0003-0454-3937}{}

\authorrunning{A. Firbas, R. Ganian, S. Meunier, and M. Nöllenburg}



\funding{Alexander Firbas, Robert Ganian, and Martin Nöllenburg are supported by the Vienna Science and Technology Fund (WWTF) [10.47379/ICT22029]. Alexander Firbas and Robert Ganian are also supported by the Austrian Science Fund (FWF) [10.55776/Y1329].}

\begin{document}

\maketitle

\begin{abstract}
Two-layer drawings of bipartite graphs place the vertices of each part on one of two parallel lines and draw the edges as straight-line links. Traditionally, the optimization goal is to find vertex permutations on one or both layers that minimize the induced number of edge crossings. This problem is \NP-hard, and crossing-minimal solutions may still contain many crossings. Recently, there has been growing interest in an orthogonal optimization goal, namely removing all crossings by vertex splitting, i.e., replacing original vertices by two or more copies and distributing the adjacencies among them.
In this paper, we study a natural extension of the two-layer vertex splitting problem in which the vertex order on one layer is constrained by a given auxiliary tree $T$, motivated by applications such as the visualization of anatomical hierarchies in the Human Reference Atlas. We investigate the parameterized complexity of this problem and obtain two main contributions: (1) a fixed-parameter algorithm with respect to the number $k$ of splits, and (2) an \ETH-tight single-exponential fixed-parameter algorithm with respect to the maximum degree of $T$. Moreover, we build on the latter result to obtain an \ETH-tight single-exponential algorithm for the classical unconstrained version of the problem, improving upon the previous $\mathcal{O}^*(2^{k\cdot \log k})$ algorithms. Finally, we also implement our algorithm and show that it performs well in practice.%

\end{abstract}

\section{Introduction}

Two-layer drawings are a classic representation of bipartite graphs that have been extensively studied from both theoretical and practical points of view.
Their popularity in graph drawing is due both to their natural use in visualizing bipartite graphs~\cite{ew-ecdbg-94} and their widespread occurrence as a subproblem in more general layered graph drawing~\cite{stt-mvuhss-81,hn-hda-13}, track layouts~\cite{dw-sqtlgs-05}, or storyline layouts~\cite{gjlm-cmsv-16}.

A nearly ubiquitous high-level aim when devising two-layer drawings is the avoidance of edge crossings. In fact,
crossing minimization in two-layer drawings with straight-line edges can be phrased as a purely combinatorial problem, asking for a permutation for each of the vertex subsets of the bipartite graph such that the number of edge crossings, defined as edge pairs whose incident vertices are inverted in the two permutations, is minimized.
It is known that testing if a planar, i.e., crossing-free, two-layer drawing exists can be done in linear time~\cite{emw-ecp-86}; these graphs are exactly forests of caterpillars or, equivalently, graphs of pathwidth at most~1.
General crossing minimization, however, is \NP-complete even if one of the permutations is fixed~\cite{ew-ecdbg-94}, and even if the graph is a tree~\cite{d-ncocmt-25}; it admits constant-factor approximations~\cite{ew-ecdbg-94,n-ibomcntd-05}, a fixed-parameter algorithm parameterized by the number of crossings~\cite{dw-efpta1cm-02,dw-efpta1cm-04}, and it has recently been the topic of the PACE challenge~\cite{KindermannKT24}.

\begin{figure}[t]
    \centering
    \includegraphics[page=1,width=\linewidth]{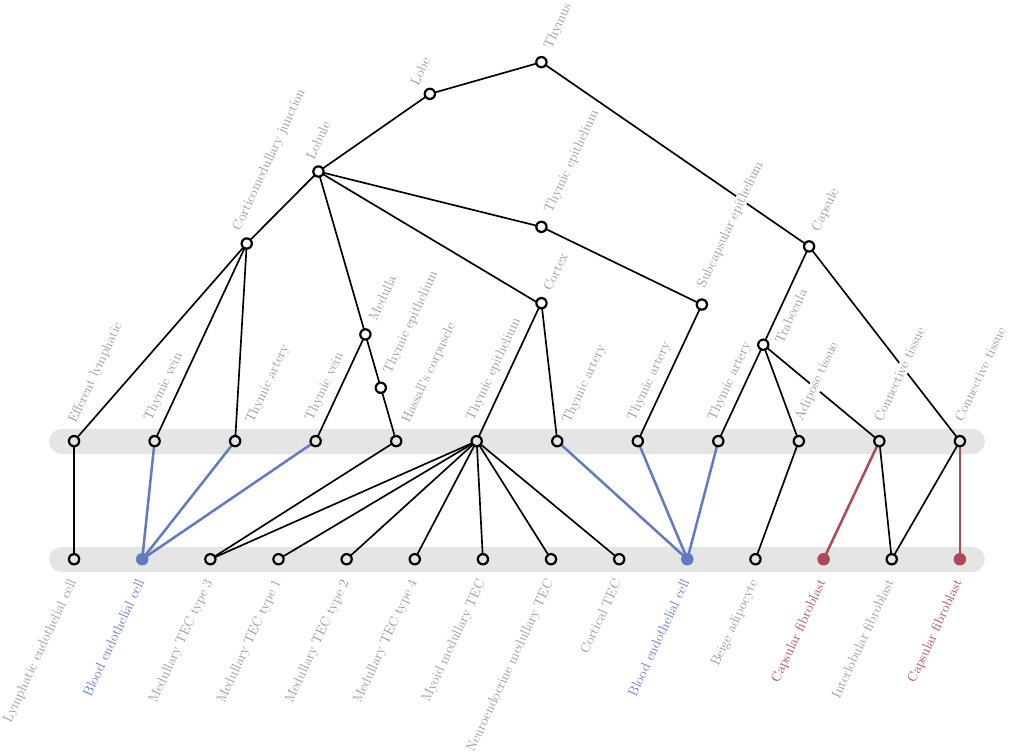}
    \caption{An illustration of an optimal solution for an instance of \prob\ modeling the Thymus organ, its sub-organs, and associated cell types (Source: \url{https://apps.humanatlas.io/kg-explorer/asct-b/thymus/v1.5}), based on our implementation described in Section~\ref{sec:implementation}. The solution uses two splits, shown in blue and red.
    }
    \label{fig:introfig}
\end{figure}

A by now well-established alternative approach for avoiding edge crossings is \emph{vertex splitting}~\cite{em-vtl-95,nstvww-pgtdvs-22,evrw-irssgdflecrtivs-23,nstvww-pgtdvs-25}. The vertex splitting operation (or \emph{split} for short) replaces a vertex $v$ of the given graph by two new, non-adjacent vertices $v'$ and $v''$ and distributes the edges originally incident to $v$ among $v'$ and $v''$, thus linking them to the former neighbors of $v$. The new vertices  jointly represent the original $v$. Since they can be placed in different positions, a single split can substantially decrease the required number of crossings of the new graph; in this sense, vertex splitting is well-suited even for two-layer drawings that may otherwise have potentially many crossings and thus suffer from poor readability in applications~\cite{jm-2scmpeha-97}. Because of this, the use of vertex splitting to avoid crossings in two-layer drawings has already been extensively studied, especially in the case where splits are only permitted on one side of the bipartition (see also the Related Work paragraph later in this section).

\subparagraph{Problem Formulation.}
While the classic two-layer vertex splitting problem admits arbitrary vertex permutations on both of its layers, here we target a variation of the problem where the admissible permutations on one side must be compatible with some planar embedding of a given tree.
Let $G=(V_t \cup V_b, E)$ be a bipartite graph and $k\in \mathbb{N}$.
Further, we are given a rooted tree $T$ whose leaf set is the \emph{constrained}, or \emph{top}, side $V_t$ of $G$.
As in the one-sided splitting setting for classical two-layer drawings~\cite{2layersplitting}, we restrict the split operation to the \emph{unconstrained}, or \emph{bottom}, side $V_b$.
We study \problong (\prob), which asks for a sequence of at most $k$ splits such that the resulting bipartite graph $G'$ admits a two-layer crossing-free drawing \emph{where the permutation~$\pi_t$ of $V_t$ must be compatible with $T$}---that is, for every internal node $v$ of $T$, the leaves of the full subtree $T_v$ rooted at $v$ must be consecutive in $\pi_t$; see \cref{fig:introfig} for an example.

Note that if $T$ is a star, there are no constraints on $\pi_t$ and our problem precisely corresponds to the classic two-layer vertex splitting problem; we denote this special case by \probStar~\cite{2layersplitting,bipartite_vertex_splitting,lxhw-pap2dvsgs-24,lxlwh-fpab1ve-23,lxlwh-fpab1ve-23-journal,povs}.
Moreover, \prob\ is solvable in linear time in the base case of $k=0$, i.e., when we simply have to find a crossing-free two-layer drawing compatible with $T$; this is a direct consequence of linear-time upward planarity testing for single-source digraphs~\cite{bbmt-ouptsd-98}.

The problem at hand is motivated by visualizing networks of anatomical
structures, cell types, and biomarkers in the context of the Human Reference
Atlas project~\cite{btqgbo-asctbhra-21}, see also
\url{https://apps.humanatlas.io/asctb-reporter/}. In this work, we focus on
the anatomical-structure--cell-type part of the ASCT+B data, where a
hierarchy of anatomical structures, called a \emph{partonomy tree}, is linked
to the different cell types of the human body. For example,
\cref{fig:introfig} shows the thymus, whose partonomy tree branches into
regions such as the capsule, lobe, lobule, cortex, and medulla, and whose
leaves include structures such as Hassall's corpuscle, thymic arteries and
veins, connective tissue, and adipose tissue. A leaf is adjacent exactly to
the cell types occurring in the corresponding structure, such as thymic
epithelial cells, endothelial cells, and fibroblasts. Thus, the visualization
should make both the partonomy tree and the cell-type composition clearly
readable. This is modeled by \prob, which asks for a planar drawing of the
tree and a crossing-free two-layer drawing after splitting a minimum number
of cell-type vertices.
We note that tree-constrained two-layer drawing problems also occur in the context of tanglegrams~\cite{bbbnosw-dcbth-09,csw-tkt-19} with phylogenetic trees constraining the permutations on both layers and in the multilayer setting of storyline layouts~\cite{gjlm-cmsv-16}.

\subparagraph{Related Work.}
Natural computational questions that have been studied for vertex splitting include determining the \emph{splitting number} of a given graph, i.e., the minimum number of splits necessary to obtain a planar drawing~\cite{ekkllm-pstg-18,ffm-snn-01}, minimizing the number of vertices to split into an arbitrary number of copies at once (also known as \emph{vertex explosions}), and minimizing the number of remaining crossings with a given budget of $k$ splits.
All of these three problems have also been considered before in the context of two-layer drawings, especially for the one-sided case, where only vertices of one side of the bipartition can be split, and all three problems are \NP-complete~\cite{2layersplitting}. In view of this, previous research---as well as our current work---attacks the vertex splitting problem through the lens of \emph{fixed-parameter tractability}~\cite{CyganFKLMPPS15,DowneyF13}.

Baumann, Pfretzschner, and Rutter~\cite{povs} studied splitting graphs to pathwidth~1, which can be seen as a generalization of the two-layer splitting problem: they provided a linear-size kernel and an $\mathcal{O}((6k+12)^k m)$-time fixed-parameter algorithm for that problem, where $m$ is the number of edges. The latter algorithm directly implies an algorithm with the same running time for the two-layer vertex splitting problem, answering an
open question of Ahmed, Kobourov, and Kryven~\cite{bipartite_vertex_splitting}. Liu, Xiao, Huang, and Wang improved this result to an $\mathcal{O}((3k+11)^k m)$-time fixed-parameter algorithm for the pathwidth-1 splitting problem and a faster $\mathcal{O}((k+5)^k m)$-time algorithm for the more restricted two-layer splitting problem~\cite{lxhw-pap2dvsgs-24}.
The related vertex explosion problem for two-layer drawings has been shown to admit a polynomial kernel, and thus is \FPT\ parameterized by the number of vertex explosions~\cite{bipartite_vertex_splitting}.
This result was first improved to a quadratic-size kernel and a single-exponential fixed-parameter algorithm~\cite{povs}, and then further to a linear-size kernel~\cite{lxlwh-fpab1ve-23,lxlwh-fpab1ve-23-journal}.

\subparagraph{Contributions.}
In the context of \prob, the previous $k^{\mathcal{O}(k)}m$-time algorithms for the two-layer splitting problem~\cite{povs,lxhw-pap2dvsgs-24} solve a special case of our problem of interest---specifically, the case where $T$ is a star. However, for general choices of $T$, it seems impossible to generalize or adapt the techniques employed there: at their core, these algorithms attack the graph-theoretic problem of splitting to a graph of pathwidth~$1$, and it is difficult to reconcile that perspective with the drawing/permutation constraints imposed by $T$. In fact, it is not even immediately obvious why \prob\ should be solvable in polynomial time for every fixed value of $k$ (i.e., lie in the complexity class \XP): while one can employ trivial branching to determine which vertices of $V_b$ need to be split, it is not possible to exhaustively branch in order to determine what these splits should look like.

As our \textbf{first contribution}, we obtain a fixed-parameter algorithm for \prob. Our result relies on a new approach that dynamically Turing-reduces the problem to the task of finding a minimum-cost permutation of tuples with constraints---a problem we call \wtsLong---and then establishing the fixed-parameter tractability of this problem. The latter then follows via a non-trivial reduction to a graph problem for which we show that the constructed graph $H$ must have pathwidth at most $\mathcal{O}(k)$, in turn allowing us to solve it via dynamic programming. In fact, we prove that $H$ must not only have bounded pathwidth but also must have $\mathcal{O}(k)$ vertex deletion distance to an interval graph---a property which allows us to leverage the state-of-the-art algorithm of Cao~\cite{cao2016linear} in order to obtain a running time of $\mathcal{O}^*(2^{k\cdot \log k})$.\footnote{Here, $\mathcal{O}^*$ notation suppresses polynomial dependencies on the number $n$ of vertices and $m$ of edges.}

\begin{restatable}{theorem}{thmsplits}
\label{thm:thmsplits}
\prob\ is fixed-parameter tractable with respect to the number $k$ of splits, and in particular can be solved in time $\mathcal{O}(2^{\mathcal{O}(k\cdot \log k)}\cdot n+m)$.
\end{restatable}

While the preceding works on two-layer splitting target $k$ as the most natural parameterization for the problem and Theorem~\ref{thm:thmsplits} follows that line, the tree-constrained setting also allows for a second perspective which can be employed even if the number $k$ of required splits is large. In particular, the maximum degree $\Delta$ in $T$ can be assumed to be relatively low, among others, in the partonomy tree setting (see, e.g., Figure~\ref{fig:large} in Section~\ref{sec:implementation}).
This raises the question of whether $\Delta$ can be used as an alternative parameter toward tractability. As our \textbf{second contribution}, we provide a positive answer to this:

\begin{restatable}{theorem}{thmdegree}
\label{thm:thmdegree}
\prob\ is fixed-parameter tractable with respect to the maximum degree $\Delta$ of $T$, and in particular can be solved in time
$\mathcal{O}\bigl(2^\Delta \cdot \Delta^3 \cdot k^3 \cdot n + m\bigr)$.
\end{restatable}

We remark that, in contrast to all fixed-parameter algorithms mentioned up to this point, the single-exponential parameter dependence achieved in Theorem~\ref{thm:thmdegree} is optimal under the Exponential Time Hypothesis (\ETH)~\cite{IMPAGLIAZZO2001512}, as follows from a known reduction from \textsc{Hamiltonian Path}~\cite{2layersplitting} (see Lemma~\ref{lem:eth-lower-probstar}).
The algorithm also leverages the aforementioned connection to \wtsLong, but employs Held--Karp-style dynamic programming over subsets~\cite{HeldK61} in the last step.

As our \textbf{third contribution}, we show that Theorem~\ref{thm:thmdegree} can be leveraged to improve the state of the art for \probStar, the classical unconstrained two-layer vertex splitting problem considered in previous works. In particular, we adapt Baumann, Pfretzschner, and Rutter's linear kernel for splitting to pathwidth~1~\cite{povs} to obtain a linear kernel for the two-layer vertex splitting problem (both with respect to $k$). By applying Theorem~\ref{thm:thmdegree} on this kernel, we obtain an \ETH-tight (cf.\ \cref{lem:eth-lower-probstar}) single-exponential algorithm:

\begin{restatable}{theorem}{thmclassic}
\label{thm:thmclassic}
\probStar{} admits an $\mathcal{O}\bigl(2^{17k} \cdot k^7 + (n+m)^2\bigr)$-time algorithm.
\end{restatable}

Finally, while the focus of our work is the complexity-theoretic foundations of the problem, as an \textbf{additional contribution} we show that the algorithm underlying Theorem~\ref{thm:thmdegree} is efficient enough to be interesting in practice. In particular, we provide an implementation and experiments~\cite{supplement} (\cref{sec:implementation}) showing that even without careful tuning, it can efficiently handle real-world data sets spanning thousands of elements.

\section{Preliminaries}
We use $[i]$ as shorthand for the set $\{1,\dots,i\}$ and fix the minimum over integers in an empty set $Z$ as $\infty$ (typically, the latter will indicate an invalid choice of $Z$). Furthermore, we use the convention that $\infty+a=\infty$ for every integer $a$.
We assume familiarity with basic graph terminology~\cite{Diestelbook} and use $V(G)$ and $E(G)$ to denote the vertex and edge sets of a graph $G=(V,E)$, respectively. A \emph{full subtree} $T'$ of a rooted tree $T$ is the subtree rooted at some node $v$ of $T$ which contains all nodes whose unique path to the root passes through $v$. We identify each node of $T$ with the full subtree it roots, and use the same symbol for both. We use $\omega(G)$ to denote the size of a maximum clique in $G$.

A \emph{path decomposition} of a graph $G=(V,E)$ is a sequence
$\mathcal{X}=(X_1,\dots,X_q)$ of subsets of $V$ such that:
\begin{enumerate}
    \item $\bigcup_{t=1}^q X_t = V$,
    \item for every edge $uv\in E$, there exists $t\in [q]$ with $\{u,v\}\subseteq X_t$,
    \item for every vertex $v \in V$, the set $\{\, t \in [q] \mid v \in X_t \,\}$ forms an interval of $[q]$.
\end{enumerate}
A path decomposition is \emph{nice} if it also satisfies the following two conditions:
\begin{enumerate}
  \setcounter{enumi}{3}
\item $X_1=X_q=\emptyset$,
\item for each $2\leq i\leq q$, $X_i$ is obtained by either \emph{introducing} or \emph{forgetting} a vertex from $X_{i-1}$.
\end{enumerate}
The \emph{width} of $\mathcal{X}$ is $\max_{t \in [q]} |X_t|-1$, and the \emph{pathwidth} of $G$,
denoted $\operatorname{pw}(G)$, is the minimum width of a path decomposition of $G$. It is well-known that a path decomposition can be transformed into a nice one of the same width in linear time.

\subsection{Problem-Specific Definitions}
Let $G$ be a bipartite graph with bipartition $V_b(G)$ and $V_t(G)$; when the graph is clear from context, we use $V_b$ and $V_t$ for brevity. A two-layer drawing of $G$ is a straight-line drawing which places $V_t$ and $V_b$ on two parallel horizontal lines; combinatorially, it is equivalent to two total orders $\pi_t$ on the ``top'' vertices $V_t$ and $\pi_b$ on the ``bottom'' vertices $V_b$.
We say that two edges $us,vt\in E$ \emph{cross} if
$u\prec_{\pi_t} v$ and $t\prec_{\pi_b} s$.
The drawing is \emph{crossing-free} if no two edges cross.
For a vertex $v\in V(G)$, $N_G(v)$ denotes its neighborhood in $G$; for a set $S\subseteq V(G)$, set $N_G(S)\coloneqq \bigcup_{v\in S} N_G(v)$.

\emph{Splitting a vertex $v$} in a bipartite graph $G$ means producing a new bipartite graph $G'$ by replacing $v$ with two non-adjacent vertices $v_1,v_2$ so that $N_{G'}(v_1), N_{G'}(v_2)$ forms a partition\footnote{We remark that in the literature, vertex splits are sometimes allowed to produce vertices with non-disjoint neighborhoods (see e.g.~\cite{cluster_editing_vertex_splitting,splitting_firbas_sorge})---however, disjoint neighborhoods can be assumed w.l.o.g.\ in the context of crossing minimization.} of $N_G(v)$.
We call $v$ the \emph{ancestor}  of $v_1,v_2$ and $v_1,v_2$ \emph{descendants} of $v$; for a sequence of splits, ancestors and descendants are considered in the reflexive and transitive sense.
Let $T$ be a rooted tree with leaf set $V_t$. For any full subtree $X$ of $T$, let $L(X)\subseteq V_t$ be the leaves of~$X$ and set
\[
B_G(X)\coloneqq N_G(L(X))\subseteq V_b.
\]
That is, $B_G(X)$ is the set of bottom vertices adjacent to leaves of $X$. See~\cref{fig:basic_notation}.

\begin{figure}[!b]
    \centering
    \includegraphics[page=3]{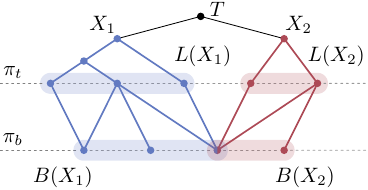}
    \caption{A solution of an instance of \prob. Full subtrees $X_1, X_2$ of $T$ and their leaves $L(\cdot)$ and bottom sets $B(\cdot)$ (with subscripts omitted for brevity) are marked in blue and red, respectively. Observe that the bottom sets appear contiguously along $\pi_b$ and intersect in at most one vertex---this key property of solutions is formalized in \cref{lem:consecutive}.}
    \label{fig:basic_notation}
\end{figure}

Given a rooted tree $T$, we recursively define for each full subtree $X$ of $T$ a set of \emph{admissible top orders} $\mathrm{Adm}_t(X)$ as follows: if $X$ is a leaf, then $\mathrm{Adm}_t(X)=\{(X)\}$; if $X$ has children $X_1,\dots,X_d$, then $\mathrm{Adm}_t(X)$ consists of all total orders on $L(X)$ which can be obtained by first choosing a permutation $(p_1,\dots,p_d)$ of $[d]$ and then, in that order, writing an order from $\mathrm{Adm}_t(X_{p_1})$ followed by an order from $\mathrm{Adm}_t(X_{p_2})$, and so on.

\problembox{\problong~(\prob)}
{A bipartite graph $G=(V_t\cup V_b,E)$, a rooted tree $T$ with leaves $V_t$, an integer $k\ge 0$.}
{Compute (1) a sequence $\zeta$ of at most $k$ vertex splits applied to $V_b(G)$ producing a graph $G'$, (2) a total order $\pi_t\in \mathrm{Adm}_t(T)$, and (3) a total order $\pi_b$ on $V_b(G')$, such that the two-layer drawing of $G'$ with $(\pi_t,\pi_b)$ is crossing-free---or, correctly determine that this is not possible.}

\looseness=-1
We call a triple $(\zeta, \pi_t, \pi_b)$ witnessing that an instance $(G,T,k)$ is a YES-instance a \emph{solution}.
Throughout the paper, we assume that $G$ has $n$ vertices and $m$ edges, contains no isolated vertices, none of the splits producing $G'$ create isolated vertices, and that $T$ has no internal nodes with a single child. This is without loss of generality, as removing an isolated vertex from $G$ preserves the (non-)existence of a solution (regardless of how it interacts with $T$) and the same also holds for contracting the edge connecting a node of $T$ to its single child in $T$.

Crucially, two full subtrees with disjoint sets of leaves have almost disjoint bottoms in $G'$:

\begin{lemma}[Properties of Solutions]

\label{lem:consecutive}
Let $(G,T,k)$ be an instance of \prob with a solution $(\zeta, \pi_t, \pi_b)$ which produces a graph $G'$.
Then:
\begin{enumerate}
    \item For every full subtree $X$ of $T$, the set $B_{G'}(X)$ appears contiguously along $\pi_b$.
    \item If $X_1, X_2$ are full subtrees of $T$ with $L(X_1)\cap L(X_2)=\emptyset$, and $I_i$ denotes the interval of $\pi_b$ spanned by $B_{G'}(X_i)$, then $I_1\cap I_2$ has size at most $1$, and any common vertex is an endpoint of both $I_1$ and $I_2$.
\end{enumerate}
\end{lemma}

\begin{proof}
By the definition of $\mathrm{Adm}_t(T)$, the leaves $L(X)$ of any full subtree $X$ must appear contiguously in $\pi_t$.

\proofsubparagraph*{Statement 1.}
Assume for contradiction that $B_{G'}(X)$ is not contiguous in $\pi_b$. Then there exist $s, t \in B_{G'}(X)$ and $m \in V_b(G') \setminus B_{G'}(X)$ such that $s \prec_{\pi_b} m \prec_{\pi_b} t$.

Because $m$ is not isolated, there exists $vm \in E(G')$ for some $v \in V_t$. Since $m \notin B_{G'}(X)$, it follows that $v \notin L(X)$. Because $L(X)$ is contiguous in $\pi_t$, $v$ must lie entirely before or entirely after $L(X)$. Assume w.l.o.g.\ that $v \prec_{\pi_t} L(X)$.

Since $s \in B_{G'}(X)$, we can choose $u \in L(X)$ such that $us \in E(G')$. We have $v \prec_{\pi_t} u$ and $s \prec_{\pi_b} m$. Thus the edges $vm$ and $us$ cross, contradicting that the drawing is crossing-free. Therefore, $B_{G'}(X)$ is contiguous in~$\pi_b$.

\proofsubparagraph*{Statement 2.}
By Statement 1, the intervals $I_1$ and $I_2$ contain exactly the vertices of $B_{G'}(X_1)$ and $B_{G'}(X_2)$. Since $L(X_1) \cap L(X_2) = \emptyset$, $L(X_1)$ and $L(X_2)$ are disjoint contiguous blocks in $\pi_t$. Assume w.l.o.g.\ that $L(X_1) \prec_{\pi_t} L(X_2)$.

First, towards a contradiction, assume  that $|I_1 \cap I_2| \ge 2$. Then there exist $s, t \in I_1 \cap I_2$ such that $s \prec_{\pi_b} t$. Because $t \in I_1$ and $s \in I_2$, we can choose $u_1 \in L(X_1)$ and $u_2 \in L(X_2)$ such that $u_1t, u_2s \in E(G')$. Since $u_1 \prec_{\pi_t} u_2$ and $s \prec_{\pi_b} t$, the edges cross, a contradiction. Thus, $|I_1 \cap I_2| \le 1$.

Second, suppose $I_1 \cap I_2 = \{v\}$. Towards a contradiction, assume that $v$ is not the rightmost element of $I_1$ with respect to $\pi_b$. Then there exists $t \in I_1$ such that $v \prec_{\pi_b} t$. Choose $u_1 \in L(X_1)$ and $u_2 \in L(X_2)$ such that $u_1t, u_2v \in E(G')$. Since $u_1 \prec_{\pi_t} u_2$ and $v \prec_{\pi_b} t$, the edges cross, a contradiction. Therefore, $v$ must be the rightmost endpoint of $I_1$.

By a symmetric argument, $v$ must be the leftmost endpoint of $I_2$.
\end{proof}

Finally, the reduction of Ahmed et al.~\cite{2layersplitting} together with the \ETH-hardness of constant-degree \textsc{Hamiltonian Path}~\cite{BjorklundHKK18} immediately yields the following lower bounds.

\begin{lemma}[ETH lower bounds inherited from {\normalfont\textsc{Hamiltonian Path}}]

\label{lem:eth-lower-probstar}
Unless the \ETH fails, \probStar admits neither a $2^{o(|V_t|)}\cdot (n+m)^{\mathcal O(1)}$-time algorithm nor a $2^{o(k)}\cdot (n+m)^{\mathcal O(1)}$-time algorithm.
\end{lemma}

\begin{proof}
\textsc{Hamiltonian Path} on $n$-vertex graphs of constant maximum degree has no $2^{o(n)}$-time algorithm unless \ETH fails~\cite{BjorklundHKK18}. The reduction of Ahmed et al.~\cite{2layersplitting} maps an instance $H=(V,E)$ of \textsc{Hamiltonian Path} with $|V|=n$ and $|E|=m$ to an equivalent \probStar instance $(G=(V_t\cup V_b,E'),k)$ with $|V_t|=n$, $|V_b|=m$, and $k=m-n+1$. On constant-degree instances, $m=\mathcal O(n)$, so $k=\mathcal O(n)$. Either claimed algorithm for \probStar would therefore give a $2^{o(n)}$-time algorithm for constant-degree \textsc{Hamiltonian Path}, a contradiction.
\end{proof}

\section{A Recurrence-Based Perspective on Vertex Split Minimization}

As an initial step towards obtaining both our algorithms for \prob, we will reformulate the decisions that need to be made at each node of $T$ into an instance of a problem we call \wtsLong. But before we can establish this correspondence, we first need to define an ``alternative objective'' for \prob.

\subsection[The opt Function]{The {\boldmath$\opt$} Function}
Both algorithms proceed bottom-up over $T$. For each full subtree $X$ of $T$ and $v_L, v_R \in B_G(X)$, the function $\opt_X(v_L, v_R)$ defined below equals the minimum number of bottom vertices adjacent to the leaves of $X$ in the resulting split graph, taken over all valid solutions whose leftmost such bottom vertex is a descendant of $v_L$ and whose rightmost is a descendant of $v_R$ ($\infty$ if no such solution exists). Since each split adds one bottom vertex, $\min_{v_L, v_R \in B_G(T)} \opt_T(v_L,v_R) - |B_G(T)|$ then equals the minimum number of splits. At an internal node $X$, computing $\opt_X$ amounts to choosing a solution for each child together with an ordering of the children, so as to maximize the number of shared endpoint vertices between adjacent children's bottom intervals (cf.\ \cref{lem:consecutive}).

\begin{definition}[Definition of $\opt$]\label{def:opt_star}
    Fix an instance $(G,T,k)$ of \prob.
    For each full subtree $X$ of $T$, define a mapping $\opt_X:\ B_G(X)\times B_G(X)\ \to\ \mathbb{N}\cup\{\infty\}$, specified recursively as follows.
    If $X$ is a leaf~$r$, we set:
    \begin{equation*}
    \opt_X(v_L, v_R) \coloneqq
    \begin{cases}
    |N_G(r)|, & \text{if } v_L \neq v_R,\\
    1, & \text{if } v_L = v_R \text{ and } |N_G(r)|=1,\\
    \infty, & \text{if } v_L = v_R \text{ and } |N_G(r)|>1.
    \end{cases}
    \end{equation*}
    If $X$ has children $X_1,\dots,X_d$, we set:
    \begin{equation*}
    \begin{aligned}
    \opt_X(v_L, v_R) \coloneqq
    \min \big\{\, & \left( \sum_{i=1}^{d} \opt_{X_i}\!\big(v_L^{(i)}, v_R^{(i)}\big) \right)
    \;-\; \bigl|\{\,j \in \{1,\dots,d-1\} : v_R^{(p_j)} = v_L^{(p_{j+1})}\,\}\bigr| \ : \\
    & (p_1,\dots,p_d)\ \text{is a permutation of } [d], \\
    & (v_L^{(i)}, v_R^{(i)}) \in B_G(X_i)\times B_G(X_i)\ \text{for all } i \in [d], \\
    & v_L^{(p_1)} = v_L,\ v_R^{(p_d)} = v_R \,\big\}.
    \end{aligned}
    \end{equation*}
\end{definition}

Before directly linking $\opt$ to \prob, we first introduce an intermediate measure $\opt^\star$ that formalizes the intended meaning of $\opt$ given above.

\begin{longdefinition}[Definition of $\opt^\star$]\label{def:opt}
    Let $(G,T,k)$ be an instance, let $X$ be a full subtree of $T$, and let $v_L,v_R\in B_G(X)$.
    We set $\opt^\star_X(v_L,v_R)$ to be the minimum of $|B_{H_X}(X)|$ over all graphs $H_X$ obtained from $G_X \coloneqq G[L(X)\cup B_G(X)]$ by a sequence of vertex splits applied only to descendants of vertices in $B_G(X)$, and over all total orders $\pi_t\in\mathrm{Adm}_t(X)$ and $\pi_b$ on $B_{H_X}(X)$, such that the two-layer drawing of $H_X$ with $(\pi_t,\pi_b)$ is crossing-free and $\pi_b$ begins with a descendant of $v_L$ and ends with a descendant of $v_R$; if no such graph and orders exist, then $\opt^\star_X(v_L,v_R)\coloneqq\infty$.
\end{longdefinition}

We now directly link $\opt$ to \prob\ by showing that it is equal to $\opt^\star$.

\begin{longlemma}[Technical Lemma]\label{lem:helper}
Let $\mathcal{I} = (I_1, \dots, I_m)$ be a finite sequence of nonempty intervals of $[n]$, where
\begin{enumerate}
    \item[(i)] $\bigcup_{i=1}^m I_i = [n]$, and
    \item[(ii)] for any distinct $a, b \in \{1, \dots, m\}$, $|I_a \cap I_b| \le 1$, and furthermore, if $I_a \cap I_b = \{x\}$, then $x = \max(I_a) = \min(I_b)$ or $x = \max(I_b) = \min(I_a)$.
\end{enumerate}
Then there is a permutation $\sigma \in S_m$ such that $\min(I_{\sigma(1)})=1$, $\max(I_{\sigma(m)})=n$, and
\begin{equation*}
n \ge
\sum_{i=1}^m |I_{\sigma(i)}|
-
\bigl|\{\,1 \le i < m : \max(I_{\sigma(i)}) = \min(I_{\sigma(i+1)})\,\}\bigr|.
\end{equation*}
Consequently,
\begin{equation*}
n \ge \min_{\sigma \in S_m} \left( \sum_{i=1}^m |I_{\sigma(i)}| - \bigl|\{\,1 \le i < m : \max(I_{\sigma(i)}) = \min(I_{\sigma(i+1)})\,\}\bigr| \right),
\end{equation*}
where $S_m$ is the symmetric group on $m$ elements.
\end{longlemma}

\begin{proof}
Order the intervals by increasing minimum, breaking ties by increasing cardinality, and let $\sigma$ be the resulting permutation. Observe that the first interval starts at $1$, and the last interval ends at $n$.

For $x\in[n]$, let $q_x$ be the number of intervals containing $x$. Observe that, if $q_x>1$, then no interval contains $x$ as an interior point; otherwise its intersection with any other interval containing $x$ would violate (ii). Hence all intervals containing $x$ consist of at most one interval ending at $x$, all singleton intervals $\{x\}$, and at most one interval starting at $x$. In the chosen order, they appear consecutively and contribute $q_x-1$ adjacent pairs with $\max(I_{\sigma(i)})=\min(I_{\sigma(i+1)})=x$.

Observe that adjacent pairs counted for distinct values of $x$ are distinct. Therefore the number of adjacent pairs with $\max(I_{\sigma(i)})=\min(I_{\sigma(i+1)})$ is at least $\sum_{x=1}^n(q_x-1)$. Since $\sum_{x=1}^n q_x=\sum_{i=1}^m |I_i|$, we obtain
\begin{equation*}
\sum_{i=1}^m |I_{\sigma(i)}|
-
\bigl|\{\,1 \le i < m : \max(I_{\sigma(i)}) = \min(I_{\sigma(i+1)})\,\}\bigr|
\le
\sum_{x=1}^n q_x-\sum_{x=1}^n(q_x-1)
=
n. \qedhere
\end{equation*}
\end{proof}

\begin{lemma}

\label{lem:opt_equals_optstar}
    An instance $(G,T,k)$ of \prob\ is positive if and only if
    there are $v_L, v_R \in B_G(T)$ with $\opt_T(v_L,v_R)-|B_G(T)| \leq k$.
\end{lemma}

\begin{proof}
By \cref{def:opt}, the instance $(G,T,k)$ is positive if and only if there exist $v_L, v_R \in B_G(T)$ with $\opt^\star_T(v_L,v_R)-|B_G(T)| \leq k$. We establish the lemma by showing $\opt^\star_X = \opt_X$ for every full subtree $X$ of $T$ by structural induction.

\proofsubparagraph*{Base Case.} Let $X$ be a leaf $r \in V_t$, and fix $v_L,v_R\in B_G(X)=N_G(r)$.

\emph{Case 1: $v_L \neq v_R$.}
For every graph $H_X$ satisfying the conditions in \cref{def:opt}, we have $|B_{H_X}(X)|\ge |N_G(r)|$.
We can achieve equality  without any split by ordering the vertices of $N_G(r)$ so that $v_L$ is first and $v_R$ is last. Since $r$ is the only top vertex, the drawing is crossing-free. Thus, $\opt^\star_X(v_L,v_R)=|N_G(r)|=\opt_X(v_L,v_R)$.

\emph{Case 2: $v_L = v_R$ and $|N_G(r)|=1$.}
Observe that the graph with no split satisfies the conditions in \cref{def:opt} and has one bottom vertex. Thus, $\opt^\star_X(v_L,v_R)=1=\opt_X(v_L,v_R)$.

\emph{Case 3: $v_L = v_R$ and $|N_G(r)|>1$.}
An order satisfying the conditions in \cref{def:opt} would need the edge $rv_L$ to be assigned to two distinct descendants of $v_L$ and thus would require a non-disjoint split, which is not allowed. Thus, $\opt^\star_X(v_L,v_R)=\infty=\opt_X(v_L,v_R)$.

\proofsubparagraph*{Inductive Step.} Let $X$ be a full subtree with children $X_1, \dots, X_d$. Assume the lemma holds for all children.

We first prove $\opt^\star_X(v_L, v_R) \le \opt_X(v_L, v_R)$. If $\opt_X(v_L,v_R)=\infty$, there is nothing to prove. Otherwise, let $(p_1, \dots, p_d)$ and $(v_L^{(i)}, v_R^{(i)})$ minimize the expression for $\opt_X(v_L, v_R)$. By the inductive hypothesis, for each $i \in [d]$, there exists a graph $H_{X_i}$ obtained from $G_{X_i}$ by splits applied only to descendants of vertices in $B_G(X_i)$ and orders $\pi_t^{(i)} \in \mathrm{Adm}_t(X_i)$ and $\pi_b^{(i)}$ on $B_{H_{X_i}}(X_i)$ witnessing
\begin{equation*}
    |B_{H_{X_i}}(X_i)| = \opt^\star_{X_i}(v_L^{(i)}, v_R^{(i)}) = \opt_{X_i}(v_L^{(i)}, v_R^{(i)}).
\end{equation*}

Take disjoint copies of the child witnesses and concatenate $\pi_b^{(p_1)}, \dots, \pi_b^{(p_d)}$ to form an order $\pi_b$. Concatenate $\pi_t^{(p_1)}, \dots, \pi_t^{(p_d)}$ to form an order $\pi_t \in \mathrm{Adm}_t(X)$. The resulting graph is obtainable from $G_X$ by splits applied only to descendants of vertices in $B_G(X)$. By assumption, each child drawing is crossing-free. Moreover, if two edges belong to different children, then their top endpoints and bottom endpoints appear in the same child-block order; hence the two edges do not cross. Thus the concatenated drawing is crossing-free.

We now identify consecutive bottom vertices that descend from the same vertex of $B_G(X)$. For each $j \in \{1, \dots, d-1\}$ with $v_R^{(p_j)} = v_L^{(p_{j+1})}$, replace the last bottom vertex of the block for $X_{p_j}$ and the first bottom vertex of the block for $X_{p_{j+1}}$ by one vertex whose neighborhood is the union of their neighborhoods. The two vertices are adjacent in $\pi_b$, so the drawing remains crossing-free. The resulting graph is still obtainable from $G_X$ by vertex splits, since all incident edges remain assigned to descendants of the same vertex of $B_G(X)$.

After applying this replacement for all such $j$, the resulting graph satisfies the conditions in \cref{def:opt} for $X,v_L,v_R$ and has
\begin{equation*}
    |B_{H_X}(X)|
    =
    \sum_{i=1}^d \opt_{X_i}(v_L^{(i)},v_R^{(i)})
    -
    \bigl|\{\,j\in\{1,\dots,d-1\}: v_R^{(p_j)}=v_L^{(p_{j+1})}\,\}\bigr|
    =
    \opt_X(v_L,v_R).
\end{equation*}
Hence $\opt^\star_X(v_L,v_R)\le \opt_X(v_L,v_R)$.

To prove $\opt^\star_X(v_L, v_R) \ge \opt_X(v_L, v_R)$, there is nothing to prove if $\opt^\star_X(v_L,v_R)=\infty$. Otherwise, let $H_X$, $\pi_t$, and $\pi_b$ witness $|B_{H_X}(X)|=\opt^\star_X(v_L,v_R)$. Set $n\coloneqq |B_{H_X}(X)|$, and identify the order $\pi_b$ with $[n]$.

For each $i\in[d]$, let $H_{X_i}$ be the subgraph of $H_X$ induced by the edges with top endpoint in $L(X_i)$, and let $I_i\subseteq[n]$ be the set of positions occupied by $B_{H_{X_i}}(X_i)$ in $\pi_b$. By \cref{lem:consecutive}, the sets $I_1,\dots,I_d$ are nonempty intervals, cover $[n]$, and satisfy the intersection condition of \cref{lem:helper}.

Apply \cref{lem:helper} to obtain a permutation $\sigma$ with $\min(I_{\sigma(1)})=1$, $\max(I_{\sigma(d)})=n$, and
\begin{equation}\label{eq:helper-use-opt-proof}
n \ge
\sum_{q=1}^d |I_{\sigma(q)}|
-
\bigl|\{\,q\in\{1,\dots,d-1\}: \max(I_{\sigma(q)})=\min(I_{\sigma(q+1)})\,\}\bigr|.
\end{equation}
We use $\sigma$ as the child order in one candidate of the recurrence for $\opt_X(v_L,v_R)$, with endpoints taken from the endpoints of the intervals $I_i$.

For each $i\in[d]$, let $v_L^{(i)}$ and $v_R^{(i)}$ be the ancestors in $B_G(X_i)$ of the bottom vertices at positions $\min(I_i)$ and $\max(I_i)$. The choice of $\sigma$ gives $v_L^{(\sigma(1))}=v_L$ and $v_R^{(\sigma(d))}=v_R$. Moreover,
\begin{equation*}
    \opt_{X_i}(v_L^{(i)},v_R^{(i)})=\opt^\star_{X_i}(v_L^{(i)},v_R^{(i)})\le |I_i|
\end{equation*}
by the induction hypothesis and by restricting the drawing to $H_{X_i}$.

Whenever $\max(I_{\sigma(q)})=\min(I_{\sigma(q+1)})$, the shared endpoint has one ancestor, so $v_R^{(\sigma(q))}=v_L^{(\sigma(q+1))}$. Therefore the recurrence for $\opt_X(v_L,v_R)$ gives
\begin{align*}
    \opt_X(v_L,v_R)
    &\le
    \sum_{q=1}^d
    \opt_{X_{\sigma(q)}}(v_L^{(\sigma(q))},v_R^{(\sigma(q))})
    -
    \left|\left\{\begin{array}{@{}c@{}}
    q\in\{1,\dots,d-1\}:\\[-.3ex]
    \max(I_{\sigma(q)})=\min(I_{\sigma(q+1)})
    \end{array}\right\}\right| \\
    &\le
    \sum_{q=1}^d |I_{\sigma(q)}|
    -
    \left|\left\{\begin{array}{@{}c@{}}
    q\in\{1,\dots,d-1\}:\\[-.3ex]
    \max(I_{\sigma(q)})=\min(I_{\sigma(q+1)})
    \end{array}\right\}\right| \\
    &\le n
    =
    \opt^\star_X(v_L,v_R),
\end{align*}
where the last inequality is \eqref{eq:helper-use-opt-proof}. This proves $\opt^\star_X(v_L,v_R)\ge \opt_X(v_L,v_R)$, hence $\opt^\star_X = \opt_X$ for every full subtree $X$ of $T$. In particular, the instance $(G,T,k)$ is positive if and only if there exist $v_L, v_R \in B_G(T)$ with $\opt_T(v_L,v_R)-|B_G(T)| \leq k$.
\end{proof}

\subsection{Weighted Tuple Sorting}
\label{sec:wts}
With the alternative perspective on our problem via $\opt$ in hand, we can reformulate \prob\ as a sorting problem of tuples.
For a 2-tuple $(a,b)$, set $\fst((a,b)) \coloneqq a$ and $\snd((a,b)) \coloneqq b$.
We further fix a distinguished symbol $\lozenge$, called the \emph{unmatchable element}.

\begin{figure}
    \centering
    \includegraphics[page=4]{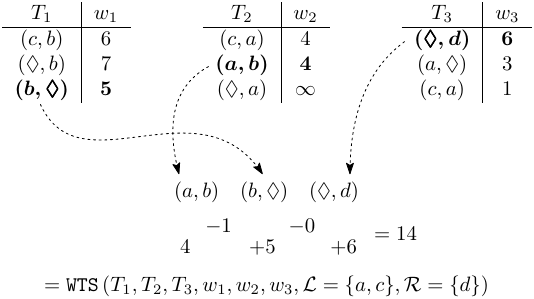}
    \caption{An instance of \wtsLong and its unique optimal solution.}
    \label{fig:wts}
\end{figure}
\problembox{\wtsLong~(\wts)}
{Finite sets $T_1, \dots, T_d$ of 2-tuples, weight functions $w_1, \ldots, w_d$ where each $w_i$ assigns each $t \in T_i$ a positive integer or $\infty$, and a \emph{left set $\mathcal L$} and a \emph{right set $\mathcal R$}.}
{Find a permutation $(p_1, \dots, p_d)$ of $[d]$ and tuples $t_i \in T_i$ for all $i \in [d]$ with $\fst(t_{p_1}) \in \mathcal L$ and $\snd(t_{p_d}) \in \mathcal R$ such that the value $\left( \sum_{i=1}^d w_i(t_i) \right) - \bigl|\{\,i \in \{1,\dots,d-1\} : \snd(t_{p_i}) = \fst(t_{p_{i+1}})\neq \lozenge\,\}\bigr|$ is minimized.}

For an instance $\mathcal{I}=(T_1,\dots,T_d, w_1,\dots,w_d, \mathcal L, \mathcal R)$ of \wtsLong, we use $\texttt{\textup{WTS}}(\mathcal{I})$ to denote its minimum value, where we set  $\texttt{\textup{WTS}}(\mathcal{I}) = \infty$ if the instance is infeasible. We also write $\texttt{\textup{WTS}}(T_1,\dots,T_d,w_1,\dots,w_d,\mathcal L,\mathcal R)$ for this value.
For each $i\in [d]$, we set $\mathcal{U}_i \coloneqq \bigl\{\, \fst(t),\ \snd(t) \ \big|\ t \in T_i \,\bigr\}$ and call $\mathcal U\coloneqq \bigcup_{i=1}^d \mathcal U_i$ the \emph{universe} of the instance. See \cref{fig:wts} for an example.

Crucially, $\opt$ can be expressed as the solution value $\texttt{\textup{WTS}}(\cdot)$ of a corresponding instance of \wts;  with Lemma~\ref{lem:opt_equals_optstar}, this yields a direct link between \wts and \prob.

\begin{observation}[Reduction to {\normalfont\wts}]\label{obs:reduction-to-wts}
Let $(G,T,k)$ be an instance of \prob{}, let $X$ be a full subtree of $T$ with children $X_1,\dots,X_c$, and let $v_L,v_R\in B_G(X)$. Then $\opt_X(v_L, v_R)$ satisfies:
\begin{equation*}
    \opt_X(v_L, v_R) = \texttt{\textup{WTS}}\bigl( B_G(X_1)^2, \dots, B_G(X_c)^2, \opt_{X_1}, \dots, \opt_{X_c}, \{v_L\}, \{v_R\} \bigr).
\end{equation*}
\end{observation}

In the remainder of this section, we obtain structural insights into \prob\ which, together with \wts, will be crucial for obtaining the efficient dynamic programming solutions presented in Sections~\ref{sec:splits} and~\ref{sec:degree}.

\subsection{Boundaries and Compression}

For a full subtree $X$ of $T$, it is possible to identify the vertices of $B_G(X)$ which also have neighbors ``outside'' of $X$. We call these vertices the boundary, formalized below.

\begin{definition}[Boundary]
Let $(G,T,k)$ be an instance of \prob, and let $X$ be a full subtree of $T$. We define the \emph{boundary} of $X$, denoted $\partial B_G(X)$, as the set of vertices in $B_G(X)$ that share an edge with a top vertex outside of $L(X)$:
\begin{equation*}
    \partial B_G(X) \coloneqq B_G(X) \cap N_G(V_t \setminus L(X)).
\end{equation*}
\end{definition}

\begin{figure}[t]
    \centering
    \includegraphics[page=6]{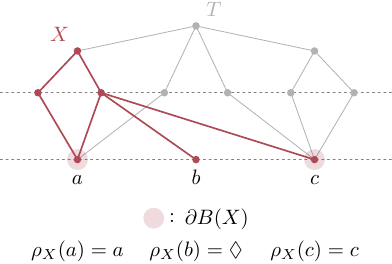}
    \caption{
    A tree $T$ with full subtree $X$, its boundary $\partial B(X)$, and $\rho(\cdot)$.}
    \label{fig:boundaries_and_projection}
\end{figure}

See also  \cref{fig:boundaries_and_projection}.

We now establish a few basic properties of boundaries, including an upper-bound on their size.

\begin{longobservation}[Properties of boundaries]\label{lem:boundary-properties}
Let $(G,T,k)$ be an instance. For every full subtree $Y$ of $T$ and every $v \in V_b$, define the \emph{internal degree} of $v$ with respect to $Y$ as
\begin{equation*}
    d_Y(v) \coloneqq |N_G(v) \cap L(Y)|.
\end{equation*}
Then, for every full subtree $X$ of $T$, the following properties hold:
\begin{enumerate}
    \item \textbf{Boundary Characterization:} For every $v \in B_G(X)$, we have $v \in \partial B_G(X)$ if and only if $d_X(v) < d_T(v)$.
    \item \textbf{Sibling Intersection:} If $X$ is not a leaf, then for any two distinct children $X_i$ and $X_j$ of $X$,
    \begin{equation*}
        B_G(X_i) \cap B_G(X_j) = \partial B_G(X_i) \cap \partial B_G(X_j).
    \end{equation*}
    \item \textbf{Boundary Inclusion:} If $X$ has children $X_1, \dots, X_c$, then
    \begin{equation*}
        \partial B_G(X) \subseteq \bigcup_{i=1}^c \partial B_G(X_i).
    \end{equation*}
    \item \textbf{Degree Additivity:} If $X$ has children $X_1, \dots, X_c$, then for every $v \in V_b$,
    \begin{equation*}
        d_X(v) = \sum_{i=1}^c d_{X_i}(v).
    \end{equation*}
\end{enumerate}
\end{longobservation}

\begin{proof}
We prove the four claims in order.

For the boundary characterization, let $v \in B_G(X)$. Since $L(T)=V_t$, we have
\begin{equation*}
    d_T(v)=d_X(v)+|N_G(v)\cap (V_t\setminus L(X))|.
\end{equation*}
Hence $d_X(v)<d_T(v)$ if and only if $v$ has a neighbor in $V_t\setminus L(X)$. Since $v\in B_G(X)$, this is equivalent to $v\in B_G(X)\cap N_G(V_t\setminus L(X))=\partial B_G(X)$.

For the sibling intersection property, the inclusion
\begin{equation*}
    \partial B_G(X_i) \cap \partial B_G(X_j) \subseteq B_G(X_i) \cap B_G(X_j)
\end{equation*}
is immediate from the definition of the boundary. Conversely, let $v\in B_G(X_i)\cap B_G(X_j)$. Then $v$ has a neighbor in $L(X_j)$, and since $L(X_i)$ and $L(X_j)$ are disjoint, this neighbor lies outside $L(X_i)$. Thus $v\in\partial B_G(X_i)$. Symmetrically, $v\in\partial B_G(X_j)$, proving the reverse inclusion.

For boundary inclusion, let $u\in\partial B_G(X)$. Since $u\in B_G(X)$ and $L(X)=\bigcup_{i=1}^c L(X_i)$, there exists $i\in[c]$ such that $u\in B_G(X_i)$. Moreover, $u$ has a neighbor in $V_t\setminus L(X)$, and this vertex also lies in $V_t\setminus L(X_i)$ because $L(X_i)\subseteq L(X)$. Hence $u\in\partial B_G(X_i)$.

Finally, degree additivity follows from the fact that the sets $L(X_1),\dots,L(X_c)$ form a partition of $L(X)$:
\begin{equation*}
    d_X(v)
    =
    |N_G(v)\cap L(X)|
    =
    \sum_{i=1}^c |N_G(v)\cap L(X_i)|
    =
    \sum_{i=1}^c d_{X_i}(v). \qedhere
\end{equation*}
\end{proof}

\begin{lemma}[Boundary Size Bound]

\label{lem:boundary-size}
Let $(G,T,k)$ be a positive instance of \prob. Then for every full subtree $X$ of $T$, the size of its boundary satisfies $|\partial B_G(X)| \le k + 2$.
\end{lemma}

\begin{proof}
    Towards a contradiction, suppose that $|\partial B_G(X)| \ge k + 3$. Since $(G, T, k)$ is a positive instance, there exists a sequence of at most $k$ vertex splits producing a graph $G'$, along with total orders $\pi_t \in \mathrm{Adm}_t(T)$ and $\pi_b$ on $V_b(G')$ such that the two-layer drawing of $G'$ with $(\pi_t,\pi_b)$ is crossing-free. Since at most $k$ original vertices of $\partial B_G(X)$ are split, $B_{G'}(X)$ must still contain at least three unsplit boundary vertices from $\partial B_G(X)$.

    By \cref{lem:consecutive}, $B_{G'}(X)$ forms a contiguous interval $I_X$ on $\pi_b$. Out of the at least three unsplit boundary vertices in $B_{G'}(X)$, at least one must lie strictly between two others in $I_X$. Let this middle vertex be $v$. Because $v \in \partial B_G(X)$, it shares an edge with some top vertex $y \in V_t \setminus L(X)$.

    Since $y$ is adjacent to $v$, the set $B_{G'}(y)$ must contain $v$. Let $I_y$ be the interval on $\pi_b$ spanned by $B_{G'}(y)$. The vertex $v$ lies in the intersection of $I_X$ and $I_y$. Since $v$ lies strictly inside $I_X$, this contradicts \cref{lem:consecutive}.

    Thus, our assumption is false, and $|\partial B_G(X)| \le k + 2$.
\end{proof}

These properties allow us to efficiently compute boundaries.

\begin{lemma}[Computing Boundaries]

\label{lem:compute-boundaries}
Let $(G,T,k)$ be an instance of \prob. One can, in time $\mathcal{O}(|E| + k|V_t|)$, either
correctly reject the instance as a negative instance, or compute the boundary sets $\partial B_G(X)$ for
all full subtrees $X$ of $T$ such that $|\partial B_G(X)| \le k+2$ for every full subtree $X$.
\end{lemma}

\begin{proof}
We process $T$ bottom-up. First, we precompute $d_T(v)=|N_G(v)|$ for all $v\in V_b$ in $\mathcal{O}(|E|)$ time. Then, for each full subtree $X$, we compute $\partial B_G(X)$ and store the value $d_X(v)$ for each $v\in\partial B_G(X)$. By \cref{lem:boundary-size}, if at any point we find a full subtree $X$ with $|\partial B_G(X)|>k+2$, then we may correctly reject the instance.

For a leaf $x\in V_t$, we have $B_G(x)=N_G(x)$. For each $v\in N_G(x)$, $d_x(v)=1$. By the Boundary Characterization property in \cref{lem:boundary-properties}, $v\in\partial B_G(x)$ if and only if $d_T(v)>1$. Thus the boundary of a leaf can be computed in time $\mathcal{O}(|N_G(x)|)$, which sums to $\mathcal{O}(|E|)$ over all leaves.

Now let $X$ be a full subtree with children $X_1,\dots,X_c$. By the Boundary Inclusion property of \cref{lem:boundary-properties},
\begin{equation*}
    \partial B_G(X) \subseteq C_X \coloneqq \bigcup_{i=1}^c \partial B_G(X_i).
\end{equation*}
Since the instance has not yet been rejected, each child boundary has size at most $k+2$.

For every $v\in C_X$, we compute $d_X(v)$ by summing all stored values $d_{X_i}(v)$ over children $X_i$ with $v\in\partial B_G(X_i)$. This accounts for all nonzero summands in
\begin{equation*}
    d_X(v)=\sum_{i=1}^c d_{X_i}(v).
\end{equation*}
Indeed, if $v\in C_X$ and $d_{X_i}(v)>0$, then $v\in B_G(X_i)$. Since $v\in C_X$, there is a child $X_j$ with $v\in\partial B_G(X_j)$. If $i\neq j$, then $v\in B_G(X_i)\cap B_G(X_j)$, and the Sibling Intersection property of \cref{lem:boundary-properties} implies $v\in\partial B_G(X_i)$. The case $i=j$ is immediate.

By the Boundary Characterization property of \cref{lem:boundary-properties}, we retain precisely those vertices $v\in C_X$ with $d_X(v)<d_T(v)$. This produces exactly $\partial B_G(X)$. If more than $k+2$ vertices are retained, we reject.

The computation for $X$ takes
\begin{equation*}
    \mathcal{O}\left(\sum_{i=1}^c |\partial B_G(X_i)|\right)=\mathcal{O}(ck)
\end{equation*}
time. Summing over all internal nodes gives $\mathcal{O}(k|V(T)|)$. Since $T$ has no internal node with exactly one child, $|V(T)|=\mathcal{O}(|V_t|)$. Together with the leaf processing and the initial degree computation, the total running time is $\mathcal{O}(|E|+k|V_t|)$.
\end{proof}

In our later dynamic programming procedures, we will be storing information about the various subsolutions occurring in a full subtree $X$. However, storing records for all possible (at most $n^2$) endpoint pairs $v_L, v_R\in B_G(X)$ of a full subtree $X$ would not be possible in fixed-parameter time. Instead, for both of our algorithmic approaches it will be advantageous to \emph{compress} the boundaries by aggregating bottom vertices with no neighbors outside of $L(X)$ into a special \emph{unmatchable element} $\lozenge$; this restricts the endpoints to $\partial B_G(X) \cup \{\lozenge\}$, a set of size at most $k+3$ by \cref{lem:boundary-size}, and both of our algorithms rely on this bound for their running times. The remainder of this section formalizes and integrates this step into our considerations.

\begin{definition}[Projection]
We define the \emph{projection mapping} $\rho_X \colon B_G(X) \cup \{\lozenge\} \to \partial B_G(X) \cup \{\lozenge\}$ as:
\begin{equation*}
    \rho_X(v) \coloneqq
    \begin{cases}
        v & \text{if } v \in \partial B_G(X), \\
        \lozenge & \text{otherwise.}
    \end{cases}
\end{equation*}
\end{definition}

Intuitively, $\rho_X$ maps all vertices from $B_G(X)$ which do not have neighbors outside of $L(X)$ to the unmatchable element; recall that the same symbol and name was also used in Subsection~\ref{sec:wts} for \wts, and the reason for this will become clear in Lemma~\ref{lem:compute-copt}.

\begin{definition}[Compressed Optimal Value]\label{def:compressed-opt}
Let $(G,T,k)$ be an instance of \prob. For any full subtree $X$ of $T$, and for any boundary elements $v_L, v_R \in \partial B_G(X) \cup \{\lozenge\}$, the \emph{compressed optimal value} $\copt_X(v_L, v_R)$ is defined as:
\begin{equation*}
    \copt_X(v_L, v_R) \coloneqq \min_{\substack{v'_L \in \rho_X^{-1}(v_L) \setminus \{\lozenge\} \\ v'_R \in \rho_X^{-1}(v_R) \setminus \{\lozenge\}}} \opt_X(v'_L, v'_R).
\end{equation*}
\end{definition}

Aside from the definition, $\copt$ can also be computed recursively via a solution to \wts (which we will be doing in both dynamic programs).

\begin{lemma}[Recurrence for $\copt_X$]

\label{lem:compute-copt}
Let $(G,T,k)$ be an instance of \prob{}. For any full subtree $X$ of $T$ with children $X_1,\dots,X_c$, let $\mathcal{U}_{X_i} \coloneqq \partial B_G(X_i) \cup \{\lozenge\}$ for each $i \in [c]$, and let $\mathcal{U}_X \coloneqq \bigcup_{i=1}^c \mathcal{U}_{X_i}$. Then, for any boundary elements $v_L, v_R \in \partial B_G(X) \cup \{\lozenge\}$, the compressed optimal value $\copt_X(v_L, v_R)$ satisfies:
\begin{equation*}
    \copt_X(v_L, v_R) = \texttt{\textup{WTS}}\bigl( \mathcal{U}_{X_1}^2, \dots, \mathcal{U}_{X_c}^2, \copt_{X_1}, \dots, \copt_{X_c}, \rho_X^{-1}(v_L) \cap \mathcal{U}_X, \rho_X^{-1}(v_R) \cap \mathcal{U}_X \bigr).
\end{equation*}
\end{lemma}

In particular, the \wts instances arising from \cref{lem:compute-copt} satisfy $\max_{i\in[c]}|\mathcal{U}_{X_i}|\le k+3$ by \cref{lem:boundary-size}, which is precisely the universe size bound required in \cref{thm:wts-fpt-pw-plus-max-universe-size}.

\begin{proof}
Set
\begin{equation*}
L \coloneqq \rho_X^{-1}(v_L)\cap \mathcal{U}_X,
\qquad
R \coloneqq \rho_X^{-1}(v_R)\cap \mathcal{U}_X.
\end{equation*}
For brevity, denote the right-hand side of the claimed equality by $F_{\mathrm c}$.
By \cref{def:compressed-opt,obs:reduction-to-wts},
\begin{equation}\label{eq:copt-wts-uncompressed}
    \copt_X(v_L, v_R)
    =
    \min_{\substack{s \in \rho_X^{-1}(v_L)\setminus\{\lozenge\}\\ t \in \rho_X^{-1}(v_R)\setminus\{\lozenge\}}}
    \texttt{\textup{WTS}}\bigl( B_G(X_1)^2, \dots, B_G(X_c)^2, \opt_{X_1}, \dots, \opt_{X_c}, \{s\}, \{t\} \bigr).
\end{equation}

For each $i\in[c]$, define
\begin{equation*}
    \psi_i \colon B_G(X_i)^2 \to \mathcal{U}_{X_i}^2,
    \qquad
    \psi_i(a,b)\coloneqq \bigl(\rho_{X_i}(a),\rho_{X_i}(b)\bigr).
\end{equation*}
For any distinct $i,j\in[c]$ and tuples $(a_i,b_i)\in B_G(X_i)^2$ and
$(a_j,b_j)\in B_G(X_j)^2$, we have
\begin{equation}\label{eq:projection-preserves-matches}
    b_i=a_j\neq\lozenge
    \iff
    \rho_{X_i}(b_i)=\rho_{X_j}(a_j)\neq\lozenge.
\end{equation}
Indeed, if $b_i=a_j=x\neq\lozenge$, then
$x\in B_G(X_i)\cap B_G(X_j)=\partial B_G(X_i)\cap \partial B_G(X_j)$ by
\cref{lem:boundary-properties}, and hence
$\rho_{X_i}(b_i)=x=\rho_{X_j}(a_j)\neq\lozenge$.
Conversely, if $\rho_{X_i}(b_i)=\rho_{X_j}(a_j)\neq\lozenge$, then both projections are equal to the
original vertices, and therefore $b_i=a_j\neq\lozenge$.

We first prove
\begin{equation}\label{eq:first-ineq-compute-copt}
    F_{\mathrm c}
    \le \copt_X(v_L, v_R).
\end{equation}
It suffices to prove that $F_{\mathrm c}$ is at most every finite value appearing in the minimum in \eqref{eq:copt-wts-uncompressed}. Fix admissible $s \in \rho_X^{-1}(v_L)\setminus\{\lozenge\}$ and $t \in \rho_X^{-1}(v_R)\setminus\{\lozenge\}$ such that the corresponding \wtsLong value is finite, and let $(p_1,\dots,p_c)$ together with tuples $(a_i,b_i)\in B_G(X_i)^2$ be an optimal solution for
\begin{equation*}
    \texttt{\textup{WTS}}\bigl( B_G(X_1)^2, \dots, B_G(X_c)^2, \opt_{X_1}, \dots, \opt_{X_c}, \{s\}, \{t\} \bigr).
\end{equation*}
For each $i\in[c]$, put $(\alpha_i,\beta_i)\coloneqq \psi_i(a_i,b_i)$.

Since $a_{p_1}=s$, we have $\alpha_{p_1}\in L$. If $s\in\partial B_G(X_{p_1})$, then
$\alpha_{p_1}=s\in \rho_X^{-1}(v_L)\cap \mathcal{U}_X=L$. Otherwise,
$\alpha_{p_1}=\lozenge$, and $s\notin\partial B_G(X)$, since any neighbor outside $L(X)$ would also be outside $L(X_{p_1})$. Hence $v_L=\rho_X(s)=\lozenge$, and again $\alpha_{p_1}\in L$. Analogously, $\beta_{p_c}\in R$.

By \eqref{eq:projection-preserves-matches}, the number of indices
$m\in\{1,\dots,c-1\}$ with $\beta_{p_m}=\alpha_{p_{m+1}}\neq\lozenge$ is equal to the number of
indices $m\in\{1,\dots,c-1\}$ with $b_{p_m}=a_{p_{m+1}}\neq\lozenge$.
Moreover, by \cref{def:compressed-opt},
\begin{equation*}
    \copt_{X_i}(\alpha_i,\beta_i)\le \opt_{X_i}(a_i,b_i)
    \qquad\text{for all } i\in[c].
\end{equation*}
Thus the projected tuple sequence is feasible for the compressed \wtsLong instance and has value at most the value of the chosen uncompressed solution. This proves~\eqref{eq:first-ineq-compute-copt}.

For the reverse inequality, if $F_{\mathrm c}=\infty$, there is nothing to prove. So assume that $F_{\mathrm c}$ is finite, and let $(p_1,\dots,p_c)$ together with
tuples $(\alpha_i,\beta_i)\in \mathcal{U}_{X_i}^2$ be an optimal solution attaining~$F_{\mathrm c}$.
For each $i\in[c]$, choose $(a_i,b_i)\in B_G(X_i)^2$ such that
\begin{equation*}
    \psi_i(a_i,b_i)=(\alpha_i,\beta_i)
    \qquad\text{and}\qquad
    \opt_{X_i}(a_i,b_i)=\copt_{X_i}(\alpha_i,\beta_i).
\end{equation*}
Such a choice exists by the definition of $\copt_{X_i}$, since the selected compressed solution has finite value and hence $\copt_{X_i}(\alpha_i,\beta_i)<\infty$ for every $i\in[c]$.

Since $\alpha_{p_1}\in L$, we have $\rho_X(\alpha_{p_1})=v_L$. If $\alpha_{p_1}\neq\lozenge$, then
$a_{p_1}=\alpha_{p_1}$ and hence $\rho_X(a_{p_1})=v_L$. If $\alpha_{p_1}=\lozenge$, then
$a_{p_1}\notin \partial B_G(X_{p_1})$, which implies $a_{p_1}\notin \partial B_G(X)$, since any neighbor outside $L(X)$ would also be outside $L(X_{p_1})$. Hence
$\rho_X(a_{p_1})=\lozenge=v_L$. Thus $a_{p_1}\in \rho_X^{-1}(v_L)\setminus\{\lozenge\}$. By the same argument, $b_{p_c}\in \rho_X^{-1}(v_R)\setminus\{\lozenge\}$.

Again by \eqref{eq:projection-preserves-matches}, the number of indices
$m\in\{1,\dots,c-1\}$ with $b_{p_m}=a_{p_{m+1}}\neq\lozenge$ is equal to the number of indices
$m\in\{1,\dots,c-1\}$ with $\beta_{p_m}=\alpha_{p_{m+1}}\neq\lozenge$.
Therefore, the value of the lifted tuple sequence
\begin{equation*}
    (a_{p_1},b_{p_1}),\dots,(a_{p_c},b_{p_c})
\end{equation*}
is exactly $F_{\mathrm c}$.

This tuple sequence is feasible for
\begin{equation*}
    \texttt{\textup{WTS}}\bigl( B_G(X_1)^2, \dots, B_G(X_c)^2, \opt_{X_1}, \dots, \opt_{X_c}, \{a_{p_1}\}, \{b_{p_c}\} \bigr).
\end{equation*}
By \cref{obs:reduction-to-wts}, this implies
\begin{equation*}
    \opt_X(a_{p_1},b_{p_c})\le F_{\mathrm c}.
\end{equation*}
Since $a_{p_1}\in \rho_X^{-1}(v_L)\setminus\{\lozenge\}$ and
$b_{p_c}\in \rho_X^{-1}(v_R)\setminus\{\lozenge\}$, the pair $(a_{p_1},b_{p_c})$ is admissible in the minimum defining $\copt_X(v_L,v_R)$. Hence
\begin{equation*}
    \copt_X(v_L,v_R)
    \le
    \opt_X(a_{p_1},b_{p_c})
    \le
    F_{\mathrm c}.
\end{equation*}
Combining this with \eqref{eq:first-ineq-compute-copt} proves the claim.
\end{proof}

\section[An Algorithm Parameterized by the Number k of Splits]{An Algorithm Parameterized by the Number \boldmath $k$ of Splits}
\label{sec:splits}

While the restructuring and insights obtained up to now will be useful throughout the paper, we now proceed towards establishing our first result:

\thmsplits*

In our first subsection, we provide

a proof of Theorem~\ref{thm:thmsplits} under the assumption that well-structured instances of \wts can be solved efficiently. To formalize this notion of ``well-structured'', we need a graph representation of \wts instances.

\begin{definition}[Interaction Graph]
The \emph{interaction graph} $H$ of a \wtsLong instance with tuple sets $T_1,\dots,T_d$ and universes $\mathcal U_1,\dots,\mathcal U_d$ is the graph
with vertex set $[d]$ in which two distinct vertices $i,j\in[d]$ are adjacent if and only if
\begin{equation*}
(\mathcal U_i\setminus\{\lozenge\})\cap(\mathcal U_j\setminus\{\lozenge\})\neq\emptyset.
\end{equation*}
\end{definition}

In particular, Theorem~\ref{thm:wts-fpt-pw-plus-max-universe-size} states that \wts can be solved efficiently when parameterized by the pathwidth of the interaction graph $H$ plus the maximum universe size.

\begin{restatable}{theorem}{thmwtspw}
\label{thm:wts-fpt-pw-plus-max-universe-size}
\wtsLong is fixed-parameter tractable parameterized by
$\operatorname{pw}(H)+\max_{i\in[d]} |\mathcal U_i|$. More precisely, let
$p\coloneqq \operatorname{pw}(H)$ and $u\coloneqq \max_{i\in[d]} |\mathcal U_i|$.
Given a nice path decomposition of $H$ of width $p$, one can solve \wtsLong in time
\begin{equation*}
\mathcal O\bigl((p+2)^3\cdot (p+1)!\cdot 2^{5p}\cdot u^{2p+4}\cdot d\bigr).
\end{equation*}
\end{restatable}

The proof of this statement is non-trivial and deferred to Subsection~\ref{sub:wtsalg}.

It relies on careful dynamic programming along a nice path decomposition and a reformulation of \wts into a connectivity problem on vertex-colored graphs.

\subsection{The Main Part}

Our aim now
is to employ Theorem~\ref{thm:wts-fpt-pw-plus-max-universe-size} to establish Theorem~\ref{thm:thmsplits}. While the required bound on the universe size is already in place, we still need to make sure that the produced instances of \wts have bounded pathwidth.
\looseness=-1
Towards this, let us use the following polynomial-time reduction rule to exhaustively remove multiple occurrences of ``simple'' children of a node in~$T$.

\begin{longredrule}[Duplicate Single-Neighbor Subtrees]\label{red:duplicate-single-neighbor}
Let $(G,T,k)$ be an instance of \prob{}, and let $X$ be a full subtree of $T$ with children
$X_1, \dots, X_c$. If there exist two distinct children $X_i$ and $X_j$ such that
$B_G(X_i) = B_G(X_j) = \{v\}$ for some vertex $v \in V_b$, apply the following modifications:
\begin{enumerate}
    \item Delete the full subtree $X_j$ from $T$.
    \item Remove the leaves $L(X_j)$ from $V_t$ and delete all their incident edges from $E$.
    \item If $c = 2$, the deletion leaves $X$ with exactly one child $X_i$. To maintain the property that internal nodes have minimum degree $2$, contract $X$ by replacing it with $X_i$ in $T$.
\end{enumerate}
\end{longredrule}

\begin{lemma}[Safety of \cref{red:duplicate-single-neighbor}]

\label{lem:safe-red}
\cref{red:duplicate-single-neighbor} is safe. That is, the resulting instance is a positive instance if and only if the original instance is a positive instance.
\end{lemma}

\begin{proof}
    Let $(G^-, T^-, k)$ be the reduced instance obtained from $(G, T, k)$ by deleting $X_j$ and $L(X_j)$.

    $(\Rightarrow)$ Suppose $(G, T, k)$ is a positive instance, witnessed by a graph $G'$ obtained via $\le k$ splits, and orders $\pi_t \in \mathrm{Adm}_t(T)$, $\pi_b$ on $V_b(G')$ yielding a crossing-free drawing. Applying the exact same sequence of splits to $G^-$ yields a valid split graph $G'^-$. Restricting $\pi_t$ to $V_t \setminus L(X_j)$ directly yields an order in $\mathrm{Adm}_t(T^-)$. Because deleting vertices and edges cannot introduce crossings, the restricted drawing with $\pi_b$ is crossing-free. Thus, $(G^-, T^-, k)$ is a positive instance.

    $(\Leftarrow)$ Suppose $(G^-, T^-, k)$ is a positive instance, witnessed by a split graph $G'^-$, and a crossing-free drawing with orders $\pi^-_t \in \mathrm{Adm}_t(T^-)$ and $\pi^-_b$ on $V_b(G'^-)$. Because $\pi^-_t \in \mathrm{Adm}_t(T^-)$, the leaves $L(X_i)$ form a contiguous block in $\pi^-_t$. Let $u \in L(X_i)$ be the rightmost leaf in this block, and let $v^* \in V_b(G'^-)$ be the specific descendant of $v$ adjacent to $u$.

    We construct a solution for $(G, T, k)$ by applying the exact same sequence of splits to $G$ to form $G'$. When splitting $v$, we partition its neighborhood such that all leaves in $L(X_j)$ are adjacent to $v^*$. We construct $\pi_t \in \mathrm{Adm}_t(T)$ by taking $\pi^-_t$ and inserting an arbitrary order from $\mathrm{Adm}_t(X_j)$ immediately after the block $L(X_i)$. This operation is admissible because $X_i$ and $X_j$ are children of $X$ in $T$. Setting $\pi_b = \pi^-_b$ completes the drawing of $G'$.

    Because $N_{G'}(L(X_j)) = \{v^*\}$ and the block $L(X_j)$ is placed immediately adjacent to $u$ in $\pi_t$ (which is also adjacent to $v^*$), no top vertices interleave them. Consequently, no edges incident to $L(X_j)$ can cross any existing edges in the drawing. Thus, $(G, T, k)$ is a positive instance.
\end{proof}

\looseness=-1
Next, we will show that after the exhaustive application of \cref{red:duplicate-single-neighbor}, the interactions between subtrees in relevant instances must be ``well-structured''; this will later allow us to generate instances of \wts of bounded pathwidth, for which Theorem~\ref{thm:wts-fpt-pw-plus-max-universe-size} can be called. Towards formalizing this, we introduce an auxiliary graph capturing how subtrees interact with each other.

\begin{definition}[Subtree Intersection Graph]\label{def:intersection-graph}
For every full subtree $X$ of $T$ with children $X_1, \dots, X_c$, we define the
\emph{subtree intersection graph} $\mathcal{I}_G(X)$ as the graph with vertex set $[c]$ in which
two distinct vertices $i,j \in [c]$ are adjacent if and only if
$
    B_G(X_i) \cap B_G(X_j) \neq \emptyset.
$
\end{definition}

We note that the subtree intersection graph is closely related to the interaction graph introduced at the beginning of this section:

\begin{observation}[Intersection Equivalence]\label{obs:intersection-equiv}
Let $(G,T,k)$ be an instance of \prob{}, and let $X$ be a full subtree of $T$ with children
$X_1, \dots, X_c$.

For any pair of distinct children $X_i$ and $X_j$, their boundaries intersect
if and only if their bottom vertex sets intersect:
\begin{equation*}
    \partial B_G(X_i) \cap \partial B_G(X_j) \neq \emptyset
    \iff
    B_G(X_i) \cap B_G(X_j) \neq \emptyset.
\end{equation*}
Consequently, the

interaction graph $H$ of the \wts instance obtained from
\cref{lem:compute-copt} in the child order $X_1,\dots,X_c$ equals
the subtree intersection graph $\mathcal{I}_G(X)$.
\end{observation}

Instead of showing that the subtree intersection graph has bounded pathwidth,
we establish an even stronger property which will be important for achieving
our running time bound (see \cref{fig:interval_structure}): they must be
``structurally similar'' to interval graphs with small cliques.

\begin{figure}
    \centering
    \includegraphics[page=5]{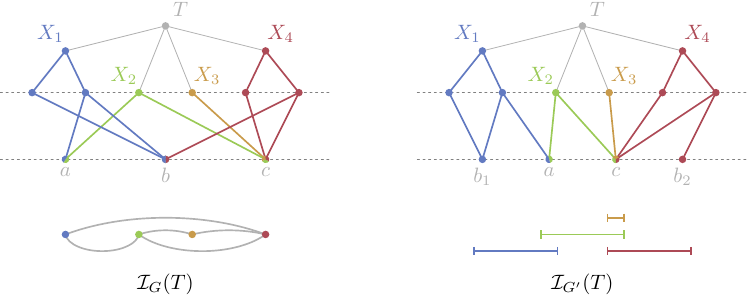}
\caption{Left: An instance of \prob\ with graph $G$ and tree $T$, with
subtree intersection graph $\mathcal{I}_G(T)$ shown below; note that
\cref{red:duplicate-single-neighbor} is not applicable. Right: A solution
where $b$ is split into $b_1, b_2$, producing $G'$, with $\mathcal{I}_{G'}(T)$
shown via its interval representation, the key structural property underlying
\cref{lem:interval-yes-instance}.}
    \label{fig:interval_structure}
\end{figure}

\begin{lemma}[Interval structure of positive instances]

\label{lem:interval-yes-instance}
Let $(G,T,k)$ be a positive instance of \prob{} exhaustively reduced under
\cref{red:duplicate-single-neighbor}. For every full subtree $X$ of $T$ with children
$X_1, \dots, X_c$, there exists a set $\mathcal{C} \subseteq [c]$ of size at most $6k$ such that
$\mathcal{I}_G(X)-\mathcal{C}$ is an interval graph with $\omega(\mathcal{I}_G(X)-\mathcal{C}) \le 3$.
\end{lemma}

\begin{proof}
Because $(G,T,k)$ is a positive instance, there exists a sequence of at most $k$ vertex splits in
$V_b$ producing a graph $G'$, along with total orders $\pi_t \in \mathrm{Adm}_t(T)$ and $\pi_b$ on
$V_b(G')$ such that the two-layer drawing of $G'$ with $(\pi_t, \pi_b)$ is crossing-free.

Let $\mathcal{I}_{G'}(X)$ be the graph with vertex set $[c]$ in which distinct vertices
$i,j \in [c]$ are adjacent if and only if
\begin{equation*}
B_{G'}(X_i) \cap B_{G'}(X_j) \neq \emptyset.
\end{equation*}
By \cref{lem:consecutive}, the sets $B_{G'}(X_i)$ form contiguous intervals on the bottom order
$\pi_b$ that overlap at most at their endpoints. Consequently, any vertex $z \in V_b(G')$ can
belong to at most two sets $B_{G'}(X_i)$ that have size at least $2$.

Additionally, there is at most one child $X_i$ with $B_{G'}(X_i)=\{z\}$. Indeed, let
$v\in V_b(G)$ be the ancestor of $z$. If $B_{G'}(X_i)=\{z\}$, then $B_G(X_i)=\{v\}$; otherwise an
edge from $L(X_i)$ to a bottom vertex with ancestor different from $v$ would still yield a bottom
vertex of $B_{G'}(X_i)$ different from $z$. Thus, two distinct children $X_i$ and $X_j$ with
$B_{G'}(X_i)=B_{G'}(X_j)=\{z\}$ would satisfy
$B_G(X_i)=B_G(X_j)=\{v\}$, contradicting the exhaustive application of
\cref{red:duplicate-single-neighbor}.

Thus, any vertex $z$ belongs to at most $2+1=3$ such sets. This implies that
$\mathcal{I}_{G'}(X)$ is an interval graph with $\omega(\mathcal{I}_{G'}(X)) \le 3$.

We now identify a deletion set corresponding to the splits. Let $S \subseteq V_b$ be the set of
original vertices that undergo a split to form $G'$. Since at most $k$ splits are performed,
$|S| \le k$. In $G'$, let $S'$ be the set of final descendants of vertices in $S$. Since each split
increases the number of descendants by one, we have $|S'| \le |S|+k \le 2k$.

Let $\mathcal{C} \subseteq [c]$ be the set of indices for children whose bottom set
in $G'$ contains at least one split descendant, that is, $B_{G'}(X_i) \cap S' \neq \emptyset.$
Because each vertex in $S'$ appears in at most $3$ sets, we have $|\mathcal{C}| \le 3|S'| \le 6k.$

For every $j \in [c] \setminus \mathcal{C}$, the set $B_G(X_j)$ contains no split vertex and hence
$B_G(X_j)=B_{G'}(X_j)$. Therefore, $\mathcal{I}_G(X)-\mathcal{C}$ is an induced subgraph of
$\mathcal{I}_{G'}(X)$. In particular, $\mathcal{I}_G(X)-\mathcal{C}$ is an interval graph. Finally, since $\omega(\mathcal{I}_{G'}(X)) \le 3$ and $\mathcal{I}_G(X)-\mathcal{C}$ is an induced subgraph of $\mathcal{I}_{G'}(X)$, we have $\omega(\mathcal{I}_G(X)-\mathcal{C}) \le 3$.
\end{proof}

To exploit \cref{lem:interval-yes-instance} algorithmically, we will apply Cao's algorithm~\cite{cao2016linear} to compute a vertex deletion set to an interval graph; the following observation implies that this computed deletion set also yields an interval graph with small cliques.

\begin{observation}

\label{lem:clique-transfer}
Let $H$ be a graph, let $S^\star \subseteq V(H)$, and let $c \in \mathbb N$.
Suppose that $\omega(H-S^\star) \le c$.
Then for every set $S \subseteq V(H)$, we have
$
    \omega(H-S) \le |S^\star| + c.
$
\end{observation}

\begin{proof}
Let $Q$ be a clique in $H-S$. Then $Q\setminus S^\star$ is a clique in $H-S^\star$, and therefore
$|Q\setminus S^\star|\le c$. Moreover, $|Q\cap S^\star|\le |S^\star|$. Hence
\begin{equation*}
    |Q|=|Q\setminus S^\star|+|Q\cap S^\star|\le c+|S^\star|.
\end{equation*}
Since this holds for every clique $Q$ in $H-S$, the claim follows.
\end{proof}

Now that we have established that subtree intersection graphs are well-structured, our aim is to exploit this to perform efficient dynamic programming. \cref{lem:interval-yes-instance} yields a bound on the pathwidth of these graphs, which is a measure that is well-suited for our intended dynamic programming routine---however, the best general-purpose algorithm for computing a width-$k$ path decomposition requires time $\mathcal{O}^*(2^{k^2})$~\cite{DBLP:conf/iwoca/Furer16} (see also~\cite{KantePT22}). In our setting, we can obtain a significantly faster running time by exploiting the structural insights from Lemma~\ref{lem:interval-yes-instance} and Cao's single-exponential algorithm for vertex deletion to interval graphs~\cite{cao2016linear}, whereas~\cref{lem:clique-transfer} allows us to use the computed vertex deletion set to obtain a bounded-width path decomposition.

\begin{lemma}[Computing path decompositions of subtree intersection graphs]

\label{lem:compute-intersection-pathdecomps}
Let $(G,T,k)$ be an instance of \prob{} exhaustively reduced under
\cref{red:duplicate-single-neighbor}, and suppose that the boundary sets
$\partial B_G(X)$ have been computed for all full subtrees $X$ of $T$ and satisfy
$|\partial B_G(X)|\le k+2$ for every full subtree $X$.
One can, in time
$\mathcal O\bigl(2^{\mathcal O(k)}\cdot |V_t|\bigr)$
either correctly reject the instance, or compute, for every non-leaf full subtree $X$ of $T$, a
nice path decomposition of $\mathcal I_G(X)$ of width at most $12k+2$.
\end{lemma}

\begin{proof}
Let $X$ be a full subtree of $T$ with children $X_1,\dots,X_c$. For every vertex
$v\in \bigcup_{i=1}^c \partial B_G(X_i)$, define
\begin{equation*}
    C_v\coloneqq \{\,i\in[c]\mid v\in \partial B_G(X_i)\,\}.
\end{equation*}
By \cref{obs:intersection-equiv}, the graph $\mathcal I_G(X)$ is obtained by making every set
$C_v$ a clique.

If $|C_v|>6k+3$ for some $v$, then $C_v$ is a clique of $\mathcal I_G(X)$ and we reject.
Indeed, in a positive reduced instance, \cref{lem:interval-yes-instance} gives a set
$S^\star$ of size at most $6k$ with $\omega(\mathcal I_G(X)-S^\star)\le 3$, and then
\cref{lem:clique-transfer} with $S=\emptyset$ gives $\omega(\mathcal I_G(X))\le 6k+3$.
Hence, after this test,
\begin{equation*}
    |E(\mathcal I_G(X))|
    \le \sum_v |C_v|^2
    \le (6k+3)\sum_v |C_v|
    \le (6k+3)c(k+2)
    =
    \mathcal O(k^2c).
\end{equation*}
Thus $\mathcal I_G(X)$ can be constructed in time $\mathcal O(k^2c)$.

Next apply Cao's $\mathcal O(6^\ell(n+m))$-time algorithm for deciding whether at most $\ell$
vertices can be deleted to obtain an interval graph~\cite{cao2016linear}, with $\ell=6k$, to either
find a set $S\subseteq V(\mathcal I_G(X))$ of size at most $6k$ such that
$\mathcal I_G(X)-S$ is an interval graph, or conclude that no such set exists. In the latter case
we reject, which is sound by \cref{lem:interval-yes-instance}. This step takes time
\begin{equation*}
    \mathcal O\bigl(6^{6k}\cdot (c+|E(\mathcal I_G(X))|)\bigr).
\end{equation*}

Otherwise, compute an interval representation of $\mathcal I_G(X)-S$ and list its maximal cliques
in their natural order. By~\cite{golumbic2004algorithmic}, these maximal cliques form a \emph{path decomposition}~\cite{CyganFKLMPPS15} of
$\mathcal I_G(X)-S$, whereas it is well-known that such a path decomposition can be transformed into a nice one in linear time. If some bag has size greater than $6k+3$, then
we reject. Indeed, in a positive reduced instance, \cref{lem:interval-yes-instance} gives a set
$S^\star$ of size at most $6k$ with $\omega(\mathcal I_G(X)-S^\star)\le 3$, and
\cref{lem:clique-transfer} applied to the current set $S$ gives
$\omega(\mathcal I_G(X)-S)\le 6k+3$.

Finally, add all vertices of $S$ to every bag. Since $|S|\le 6k$ and all remaining bags have size at
most $6k+3$, this yields a path decomposition of $\mathcal I_G(X)$ (which can once again be made nice in linear time) of width at most
\begin{equation*}
    (6k+3)+6k-1=12k+2.
\end{equation*}

For the node $X$, the running time is
\begin{equation*}
    \mathcal O(k^2c)
    +
    \mathcal O\bigl(6^{6k}\cdot (c+|E(\mathcal I_G(X))|)\bigr)
    =
    \mathcal O\bigl(2^{\mathcal O(k)}\cdot c\bigr),
\end{equation*}
where we use $|E(\mathcal I_G(X))|=\mathcal O(k^2c)$ and absorb polynomial factors in $k$ into
$2^{\mathcal O(k)}$. Since the sum of the numbers of children over all internal nodes of $T$ is
$\mathcal O(|V_t|)$, the total time is
\begin{equation*}
    \mathcal O\bigl(2^{\mathcal O(k)}\cdot |V_t|\bigr). \qedhere
\end{equation*}
\end{proof}

The final component we need to assemble a proof of \cref{thm:thmsplits} is a dynamic programming algorithm that employs the path decompositions constructed in \cref{lem:compute-intersection-pathdecomps}.

This result is technical and non-trivial; its proof relies on a reformulation of \wts into a graph problem and is deferred to Subsection~\ref{sub:wtsalg}.

\begin{restatable}{lemma}{lempwdp}

\label{lem:fpt-by-k-given-pathdecomp}
Let $(G,T,k)$ be an instance of \prob. Suppose that for every non-leaf full subtree $X$ of $T$, a path
decomposition of the subtree intersection graph $\mathcal I_G(X)$ of width at most $12k+2$ is given.
Then one can decide whether $(G,T,k)$ is a positive instance in time
\begin{equation*}
    \mathcal O\bigl(2^{\mathcal O(k\log k)}\cdot |V_t| + |E|\bigr).
\end{equation*}

\end{restatable}

With this, it is possible to establish \cref{thm:thmsplits} by exhaustively applying \cref{red:duplicate-single-neighbor}, then computing the path decompositions of the subtree intersection graphs as explained in \cref{lem:compute-intersection-pathdecomps}, and then calling \cref{lem:fpt-by-k-given-pathdecomp}.

\begin{proof}[Proof of Theorem~\ref{thm:thmsplits}]
Given an instance $(G,T,k)$, we first exhaustively apply
\cref{red:duplicate-single-neighbor} in time $\mathcal{O}(|V_t|+|E|)$. By \cref{lem:safe-red}, this preserves equivalence.

Let $(G',T',k)$ be the reduced instance. We compute the boundary sets of all full subtrees of $T'$ via
\cref{lem:compute-boundaries}. If the algorithm rejects, then $(G',T',k)$, and hence also the
original instance, is a negative instance. Otherwise, this step takes time $\mathcal{O}(|E|+k|V_t|)$.

We then compute nice path decompositions of the subtree intersection graphs via
\cref{lem:compute-intersection-pathdecomps}. If the algorithm rejects, then $(G',T',k)$, and hence
also the original instance, is a negative instance. Otherwise, for every non-leaf full subtree $X$ of $T'$
we obtain a nice path decomposition of $\mathcal I_{G'}(X)$ of width at most $12k+2$. This step takes time
\begin{equation*}
    \mathcal O\bigl(2^{\mathcal O(k)}\cdot |V_t|\bigr).
\end{equation*}

We then apply \cref{lem:fpt-by-k-given-pathdecomp} to the reduced instance $(G',T',k)$ together
with the computed nice path decompositions. This decides whether $(G',T',k)$ is a positive instance within
time
\begin{equation*}
    \mathcal O\bigl(2^{\mathcal O(k\log k)}\cdot |V_t| + |E|\bigr).
\end{equation*}

Therefore, the original instance can be decided in time
\begin{equation*}
    \mathcal{O}\bigl(2^{\mathcal{O}(k\log k)}\cdot |V_t| + |E|\bigr).
\end{equation*}
Finally, correctness follows from
\cref{lem:safe-red,lem:compute-boundaries,lem:compute-intersection-pathdecomps,lem:fpt-by-k-given-pathdecomp}.
\end{proof}

\subsection[A WTS-Based Algorithm for Theorem 1]{A {\normalfont\wts}-Based Algorithm for \cref{thm:thmsplits}}
\label{sub:wtsalg}

Our efficient fixed-parameter algorithm underlying \cref{thm:thmsplits} is strongly based on insights into the \wtsLong problem introduced in Section~\ref{sec:wts}. Throughout this subsection, let us hence fix an instance of \wtsLong consisting of tuple sets
$T_1,\dots,T_d$, weight functions $w_1,\dots,w_d$, a left set $\mathcal L$, a right set
$\mathcal R$, and the associated universes $\mathcal U_1,\dots,\mathcal U_d$ and
$\mathcal U=\bigcup_{i=1}^d \mathcal U_i$.

To present our algorithm it will be useful to introduce an auxiliary graph representation of this \wts instance. First, we define its vertex set and then the corresponding graph:
\begin{definition}[Annotated Tuples]
For every index $i\in[d]$ and every tuple $(a,b)\in T_i$, we define the corresponding
\emph{annotated tuple} $(a,b)_i$. We treat annotated tuples from different tuple sets as distinct,
that is, $(a,b)_i\neq(a,b)_j$ whenever $i\neq j$. For an annotated tuple $x=(a,b)_i$, we set
\begin{equation*}
\fst(x)\coloneqq a,\qquad \snd(x)\coloneqq b,\qquad w(x)\coloneqq w_i((a,b)).
\end{equation*}
We call $i$ the \emph{color} of $x$.
\end{definition}

\begin{definition}[Compatibility Digraph]
The \emph{compatibility digraph} $D$ of the fixed \wtsLong instance is the
directed graph with vertex set
\begin{equation*}
V(D)=\{(a,b)_i\mid i\in[d],\ (a,b)\in T_i\}.
\end{equation*}
For two annotated tuples $x,y\in V(D)$, we add the arc $(x,y)$ if and only if $x$ and $y$ have
distinct colors and
\begin{equation*}
\snd(x)=\fst(y)\neq\lozenge.
\end{equation*}
\end{definition}

For a directed path $C$, its \emph{path order} is the order of its vertices along its arcs; we
write $\first(C)$ and $\last(C)$ for the first and last vertex in this order. For a sequence
$S=(x_1,\dots,x_r)$ of annotated tuples, define
\begin{equation*}
\val(S)\coloneqq
\sum_{\ell=1}^r w(x_\ell)
-
\bigl|\{\,\ell\in[r-1]\mid \snd(x_\ell)=\fst(x_{\ell+1})\neq\lozenge\,\}\bigr|.
\end{equation*}

A \emph{directed linear forest} is a directed graph in which every weakly connected component is a
directed path. With this, we can formalize two notions that will be targeted in our reformulated graph problem:

\begin{definition}[Colorful Directed Linear Forests]
A directed linear forest $\mathcal F$ in $D$ is \emph{colorful} if no two vertices of
$V(\mathcal F)$ have the same color. We set
\begin{equation*}
\cost(\mathcal F)\coloneqq \sum_{x\in V(\mathcal F)} w(x)-|E(\mathcal F)|.
\end{equation*}
A component $C$ of $\mathcal F$ is called \emph{left-starting} if
$\fst(\first(C))\in\mathcal L$, and is called \emph{right-ending} if $\snd(\last(C))\in\mathcal R$.
\end{definition}

\begin{definition}[Feasible Families of Components]
Let $\mathcal C$ be a finite family of directed paths. We call $\mathcal C$ \emph{feasible} if its
elements can be ordered so that the first path in the order is left-starting and the last path in
the order is right-ending.
\end{definition}

With this, we can establish:

\begin{lemma}[Forest Reformulation of {\normalfont\wts}]\label{lem:wts-forest}
We have
\begin{align*}
\texttt{\textup{WTS}}(&T_1,\dots,T_d,w_1,\dots,w_d,\mathcal L,\mathcal R)\\
&=
\min\Bigl\{\,\cost(\mathcal F)\,\Bigm|\\
&\qquad \mathcal F \text{ is a colorful directed linear forest in } D,\\
&\qquad \mathcal F \text{ contains exactly one annotated tuple of each color } i\in[d],\\
&\qquad \text{and the components of } \mathcal F \text{ form a feasible family}
\Bigr\}.
\end{align*}
\end{lemma}

\begin{proof}
First let $(p_1,\dots,p_d)$ and tuples $t_i\in T_i$ be a feasible solution of the
\wtsLong instance. Let $x_i$ be the annotated tuple corresponding to $t_i$.
Construct a directed graph $\mathcal F$ on $\{x_i\mid i\in[d]\}$ by adding the arc
$(x_{p_\ell},x_{p_{\ell+1}})$ exactly for those $\ell\in[d-1]$ with
$\snd(x_{p_\ell})=\fst(x_{p_{\ell+1}})\neq\lozenge$. Since these arcs are taken from consecutive
pairs in one linear order, $\mathcal F$ is a colorful directed linear forest containing exactly one
vertex of each color, with components ordered by their first occurrence in the sequence
$x_{p_1},\dots,x_{p_d}$. The first component is left-starting because
$\fst(x_{p_1})\in\mathcal L$, and the last component is right-ending because
$\snd(x_{p_d})\in\mathcal R$. Moreover, $\cost(\mathcal F)=\val(x_{p_1},\dots,x_{p_d})$.

Conversely, let $\mathcal F$ be a colorful directed linear forest satisfying the three stated
conditions. Order its components so that the first component is left-starting and the last component
is right-ending, and inside each component list the vertices in path order. This gives a feasible
solution of the \wtsLong instance with value at most $\cost(\mathcal F)$,
because every arc of $\mathcal F$ contributes one matching bonus, and additional bonuses may only
occur between consecutive components. Hence the optimum WTS value is at most the right-hand side,
while the previous paragraph shows that the right-hand side is at most the optimum WTS value.
\end{proof}

\subparagraph{Setting up the DP.}
We can now gradually introduce the records for the dynamic programming algorithm that will solve the problem defined in Lemma~\ref{lem:wts-forest}.
Fix a nice path decomposition $\mathcal X=(X_1,\dots,X_q)$ of $H$. For every $t\in[q]$, set
\begin{equation*}
A_t\coloneqq X_t,\qquad
P_t\coloneqq \bigcup_{s=1}^t X_s,\qquad
F_t\coloneqq P_t\setminus A_t,\qquad
U_t\coloneqq [d]\setminus P_t.
\end{equation*}
We call $A_t$ the \emph{active} colors, $P_t$ the \emph{processed} colors,
$F_t$ the \emph{forgotten} colors, and $U_t$ the \emph{unprocessed} colors at node $t$. Notice that:

\begin{observation}[Active Colors Separate Past and Future]\label{obs:wts-active-separator}
Let $t\in[q]$, let $i\in F_t$, and let $j\in U_t$. Then $\{i,j\}\notin E(H)$.
\end{observation}

\begin{observation}\label{obs:wts-no-forgotten-neighbor}
Let $t\in\{2,\dots,q\}$ be an introduce node, and let $i$ be the introduced color. If $x$ is an
annotated tuple of color $i$, then in any $t$-partial solution no neighbor of $x$ in the directed
forest has color in $F_{t-1}$.
\end{observation}

The following definitions lay down the syntax of our dynamic programming records:

\begin{definition}[Partial Solutions]
Let $t\in[q]$. A \emph{$t$-partial solution} is a colorful directed linear forest
$\mathcal F$ in $D$ whose set of used colors is exactly $P_t$.
\end{definition}

\begin{definition}[Active Traces]
An \emph{active trace} is a nonempty sequence $P=(x_1,\dots,x_r)$ of distinct annotated tuples.
We set
\begin{equation*}
V(P)\coloneqq \{x_1,\dots,x_r\},\qquad
\first(P)\coloneqq x_1,\qquad
\last(P)\coloneqq x_r.
\end{equation*}
A family of active traces is \emph{vertex-disjoint} if no annotated tuple occurs in two traces.
\end{definition}

\begin{definition}[Records]
Let $t\in[q]$. A \emph{record} for bag $t$ is a tuple
\begin{equation*}
\rho=(\sel,\traces,\data,\forgotten)
\end{equation*}
satisfying the following conditions:
\begin{enumerate}
    \item $\sel$ assigns to every active color $i\in A_t$ an annotated tuple $\sel(i)$ of color $i$;
    \item $\traces$ is a family of pairwise vertex-disjoint active traces whose vertex set is
    exactly $\{\sel(i)\mid i\in A_t\}$;
    \item $\data$ is a mapping
    \begin{equation*}
        \data:\traces\to\{0,1\}^4;
    \end{equation*}
    \item $\forgotten$ is a mapping
    \begin{equation*}
        \forgotten:\{0,1\}^2\to\{0,1,2\}.
    \end{equation*}
\end{enumerate}
\end{definition}

For a record $\rho$, we write $\sel_\rho,\traces_\rho,\data_\rho,\forgotten_\rho$ for its four
components. For a trace $P\in\traces_\rho$, if
\begin{equation*}
\data_\rho(P)=(a,b,c,d),
\end{equation*}
then we set
\begin{equation*}
\begin{aligned}
\startsLeft_\rho(P)&\coloneqq a,&
\stopsRight_\rho(P)&\coloneqq b,\\
\openPred_\rho(P)&\coloneqq c,&
\openSucc_\rho(P)&\coloneqq d.
\end{aligned}
\end{equation*}

For an annotated tuple $x$, we define
\begin{equation*}
\data(x)\coloneqq (\langle\,\fst(x)\in\mathcal L\,\rangle,\langle\,\snd(x)\in\mathcal R\,\rangle,1,1).
\end{equation*}

The semantics of a record are now given by:

\begin{definition}[Realization of a Record]
Let $t\in[q]$, let $\mathcal F$ be a $t$-partial solution, and let
$\rho=(\sel,\traces,\data,\forgotten)$ be a record for bag $t$.

We say that $\mathcal F$ \emph{realizes} $\rho$ if the following conditions hold.
\begin{enumerate}
    \item \emph{Active vertices.}
    For every active color $i\in A_t$, the unique vertex of color $i$ in $\mathcal F$ equals
    $\sel(i)$.

    \item \emph{Traces induced by active vertices.}
    Let $C$ be a component of $\mathcal F$ that contains at least one active color, and let
    $x_1,\dots,x_m$ be the vertices of $C$ in path order. Let
    $x_{j_1},\dots,x_{j_s}$ be the subsequence of those vertices whose colors lie in $A_t$.
    Then the active trace $(x_{j_1},\dots,x_{j_s})$ belongs to $\traces$. Moreover, every trace in
    $\traces$ arises in this way from exactly one component of $\mathcal F$.

    \item \emph{Trace data.}
    For every trace $P\in\traces$, let $C_P$ denote the unique component of $\mathcal F$ that gives
    rise to $P$ in Item~2. Then
    \begin{equation*}
    \startsLeft_\rho(P)=\langle\,\fst(\first(C_P))\in\mathcal L\,\rangle,\qquad
    \stopsRight_\rho(P)=\langle\,\snd(\last(C_P))\in\mathcal R\,\rangle,
    \end{equation*}
    \begin{equation*}
    \openPred_\rho(P)=\langle\,\first(C_P)=\first(P)\,\rangle,\qquad
    \openSucc_\rho(P)=\langle\,\last(C_P)=\last(P)\,\rangle.
    \end{equation*}

    \item \emph{Forgotten components.}
    Let
    \begin{equation*}
    \mathcal C_t^{\mathrm{forgotten}}(\mathcal F)\coloneqq
    \{\,C \text{ component of }\mathcal F \mid
    V(C)\cap \{\sel(i)\mid i\in A_t\}=\emptyset\,\}.
    \end{equation*}
    Then for every $(a,b)\in\{0,1\}^2$,
    \begin{equation*}
\begin{aligned}
    \forgotten(a,b)
    &=
    \min\left(
    2,
    \left|
    \left\{\begin{array}{@{}c@{}}
    C\in \mathcal C_t^{\mathrm{forgotten}}(\mathcal F)\mid\\[-.3ex]
    [\fst(\first(C))\in\mathcal L]=a,\
    [\snd(\last(C))\in\mathcal R]=b
    \end{array}\right\}
    \right|
    \right).
\end{aligned}
    \end{equation*}
\end{enumerate}
\end{definition}

\begin{definition}[Record Table]
For every bag $t\in[q]$ and every record $\rho$ for bag $t$, we set
\begin{equation*}
R_t[\rho]\coloneqq
\min\{\,\cost(\mathcal F)\mid
\mathcal F \text{ is a $t$-partial solution that realizes } \rho\,\}.
\end{equation*}
\end{definition}

\subparagraph{Initialization.} We now describe how the dynamic program will be initialized:

\begin{definition}[Initial Record]
Let $\delta_0:\{0,1\}^2\to\{0,1,2\}$ be the constant-$0$ mapping. The \emph{initial record} is
\begin{equation*}
\rho_\emptyset\coloneqq (\emptyset,\emptyset,\emptyset,\delta_0),
\end{equation*}
where the third coordinate is the empty mapping with domain $\emptyset$.
\end{definition}

\begin{observation}[Initialization]\label{obs:wts-initialization}
At the first bag $X_1=\emptyset$, the only realizable record is $\rho_\emptyset$, and
\begin{equation*}
R_1[\rho_\emptyset]=0,
\qquad
R_1[\rho]=\infty \text{ for every other record } \rho.
\end{equation*}
\end{observation}

\subparagraph{Introduce Nodes.} Towards dealing with these nodes, we first define attachable traces and the corresponding operators:

\begin{definition}[Left- and Right-Attachable Traces]
Let $t\in\{2,\dots,q\}$ be an introduce node, let $i$ be the unique color with
$X_t=X_{t-1}\cup\{i\}$, let $\rho$ be a record for bag $t-1$, and let $x$ be an annotated tuple of
color $i$. We define
\begin{equation*}
\mathcal A^L_\rho(x)\coloneqq
\{\,P\in\traces_\rho \mid \openSucc_\rho(P)=1 \text{ and } (\last(P),x)\in E(D)\,\},
\end{equation*}
and
\begin{equation*}
\mathcal A^R_\rho(x)\coloneqq
\{\,P\in\traces_\rho \mid \openPred_\rho(P)=1 \text{ and } (x,\first(P))\in E(D)\,\}.
\end{equation*}
\end{definition}

\begin{definition}[Introduce Operators]
Let $t$ be an introduce node, let $i$ be the introduced color, let $\rho$ be a record for bag
$t-1$, and let $x$ be an annotated tuple of color $i$.

\subparagraph*{No attachment.}
Let the record $\operatorname{Intro}^{0}_t(\rho,x)$ be the unique record $\widehat\rho$ such that
\begin{equation*}
\sel_{\widehat\rho}=\sel_\rho\cup\{(i,x)\},\qquad
\traces_{\widehat\rho}=\traces_\rho\cup\{(x)\},\qquad
\forgotten_{\widehat\rho}=\forgotten_\rho,
\end{equation*}
every trace in $\traces_\rho$ keeps its old data, and
\begin{equation*}
\data_{\widehat\rho}((x))=\data(x).
\end{equation*}

\subparagraph*{Left attachment.}
Let $P\in \mathcal A^L_\rho(x)$, and let $P\cdot x$ denote the active trace obtained from $P$ by
appending $x$. Let the record $\operatorname{Intro}^{L}_t(\rho,x,P)$ be the unique record
$\widehat\rho$ such that
\begin{equation*}
\sel_{\widehat\rho}=\sel_\rho\cup\{(i,x)\},\qquad
\traces_{\widehat\rho}=(\traces_\rho\setminus\{P\})\cup\{P\cdot x\},\qquad
\forgotten_{\widehat\rho}=\forgotten_\rho,
\end{equation*}
every trace in $\traces_\rho\setminus\{P\}$ keeps its old data, and
\begin{equation*}
\data_{\widehat\rho}(P\cdot x)
=
\bigl(\startsLeft_\rho(P),\,\langle\,\snd(x)\in\mathcal R\,\rangle,\,\openPred_\rho(P),\,1\bigr).
\end{equation*}

\subparagraph*{Right attachment.}
Let $Q\in \mathcal A^R_\rho(x)$, and let $x\cdot Q$ denote the active trace obtained from $Q$ by
prepending $x$. Let the record $\operatorname{Intro}^{R}_t(\rho,x,Q)$ be the unique record
$\widehat\rho$ such that
\begin{equation*}
\sel_{\widehat\rho}=\sel_\rho\cup\{(i,x)\},\qquad
\traces_{\widehat\rho}=(\traces_\rho\setminus\{Q\})\cup\{x\cdot Q\},\qquad
\forgotten_{\widehat\rho}=\forgotten_\rho,
\end{equation*}
every trace in $\traces_\rho\setminus\{Q\}$ keeps its old data, and
\begin{equation*}
\data_{\widehat\rho}(x\cdot Q)
=
\bigl(\langle\,\fst(x)\in\mathcal L\,\rangle,\,\stopsRight_\rho(Q),\,1,\,\openSucc_\rho(Q)\bigr).
\end{equation*}

\subparagraph*{Left-and-right attachment.}
Let $P\in\mathcal A^L_\rho(x)$ and $Q\in\mathcal A^R_\rho(x)$ with $P\neq Q$, and let
$P\cdot x\cdot Q$ denote the active trace obtained by concatenating $P$, $x$, and $Q$.
Let the record $\operatorname{Intro}^{LR}_t(\rho,x,P,Q)$ be the unique record $\widehat\rho$ such that
\begin{equation*}
\sel_{\widehat\rho}=\sel_\rho\cup\{(i,x)\},\qquad
\traces_{\widehat\rho}=(\traces_\rho\setminus\{P,Q\})\cup\{P\cdot x\cdot Q\},\qquad
\forgotten_{\widehat\rho}=\forgotten_\rho,
\end{equation*}
every trace in $\traces_\rho\setminus\{P,Q\}$ keeps its old data, and
\begin{equation*}
\data_{\widehat\rho}(P\cdot x\cdot Q)
=
\bigl(\startsLeft_\rho(P),\,\stopsRight_\rho(Q),\,\openPred_\rho(P),\,\openSucc_\rho(Q)\bigr).
\end{equation*}
\end{definition}

\begin{lemma}[Introduce Recurrence]\label{lem:wts-introduce}
Let $t$ be an introduce node, and let $i$ be the introduced color. Then for every record
$\widehat\rho$ of bag $t$,
\begin{equation*}
R_t[\widehat\rho]
=
\min\{M_0,M_L,M_R,M_{LR}\},
\end{equation*}
where
\begin{equation*}
M_0=
\min\{\,R_{t-1}[\rho]+w(x)\mid
\widehat\rho=\operatorname{Intro}^{0}_t(\rho,x)\,\},
\end{equation*}
\begin{equation*}
M_L=
\min\{\,R_{t-1}[\rho]+w(x)-1\mid
P\in\mathcal A^L_\rho(x),\
\widehat\rho=\operatorname{Intro}^{L}_t(\rho,x,P)\,\},
\end{equation*}
\begin{equation*}
M_R=
\min\{\,R_{t-1}[\rho]+w(x)-1\mid
Q\in\mathcal A^R_\rho(x),\
\widehat\rho=\operatorname{Intro}^{R}_t(\rho,x,Q)\,\},
\end{equation*}
and
\begin{equation*}
M_{LR}=
\min\{\,R_{t-1}[\rho]+w(x)-2\mid
P\in\mathcal A^L_\rho(x),\ Q\in\mathcal A^R_\rho(x),\ P\neq Q,\
\widehat\rho=\operatorname{Intro}^{LR}_t(\rho,x,P,Q)\,\},
\end{equation*}
where each minimum ranges over all records $\rho$ of bag $t-1$ and all annotated tuples $x$ of color~$i$.
\end{lemma}

\begin{proof}
Let $i$ be the introduced color.

We first prove that every candidate counted by $M_0$, $M_L$, $M_R$, or $M_{LR}$ is attainable. Let
$\mathcal F$ be a $(t-1)$-partial solution realizing $\rho$, and let $x$ be an annotated tuple of
color $i$. Since $i\notin P_{t-1}$, the vertex $x$ does not belong to $\mathcal F$.

For $M_0$, add $x$ as an isolated vertex. Observe that the cost increases by $w(x)$, and the resulting forest
realizes $\operatorname{Intro}^{0}_t(\rho,x)$. For $M_L$, let
$P\in\mathcal A^L_\rho(x)$. Then the component inducing $P$ ends at $\last(P)$, and
$(\last(P),x)$ is an arc of $D$. Observe that adding $x$ and this arc preserves the directed-linear-forest
property, increases the cost by $w(x)-1$, and realizes $\operatorname{Intro}^{L}_t(\rho,x,P)$.
The case $M_R$ is symmetric. For $M_{LR}$, let $P\in\mathcal A^L_\rho(x)$ and
$Q\in\mathcal A^R_\rho(x)$ with $P\neq Q$. Note that the traces $P$ and $Q$ belong to distinct components.
Hence, adding $x$ together with the arcs $(\last(P),x)$ and $(x,\first(Q))$ merges these two components
into one directed path, increases the cost by $w(x)-2$, and realizes
$\operatorname{Intro}^{LR}_t(\rho,x,P,Q)$. Thus $R_t[\widehat\rho]\le \min\{M_0,M_L,M_R,M_{LR}\}$.

For the reverse inequality, let $\widehat{\mathcal F}$ be a $t$-partial solution realizing
$\widehat\rho$ with cost $R_t[\widehat\rho]$, and let $x$ be the unique vertex of color $i$ in
$\widehat{\mathcal F}$. Remove $x$ and its incident arcs to obtain a directed linear forest
$\mathcal F$ on the colors $P_{t-1}$. By \cref{obs:wts-no-forgotten-neighbor}, every neighbor of
$x$ in $\widehat{\mathcal F}$ has active color at bag $t-1$.

Note that the vertex $x$ has indegree at most one and outdegree at most one in $\widehat{\mathcal F}$. If it
has no incident arc, then $\mathcal F$ realizes a record $\rho$ with
$\widehat\rho=\operatorname{Intro}^{0}_t(\rho,x)$ and
$\cost(\widehat{\mathcal F})=\cost(\mathcal F)+w(x)$. Thus $M_0\le R_t[\widehat\rho]$.

If $x$ has only an incoming arc, then its predecessor is the last vertex of an open trace $P$ in
the record $\rho$ realized by $\mathcal F$. Hence $P\in\mathcal A^L_\rho(x)$,
$\widehat\rho=\operatorname{Intro}^{L}_t(\rho,x,P)$, and
$\cost(\widehat{\mathcal F})=\cost(\mathcal F)+w(x)-1$. Thus $M_L\le R_t[\widehat\rho]$.

If otherwise $x$ has only an outgoing arc, the symmetric argument gives a record $\rho$ and a trace
$Q\in\mathcal A^R_\rho(x)$ with
$\widehat\rho=\operatorname{Intro}^{R}_t(\rho,x,Q)$, and hence $M_R\le R_t[\widehat\rho]$.

Finally, if $x$ has both an incoming and an outgoing arc, then removing $x$ separates its component
into two components, inducing two distinct traces $P,Q$ in the record $\rho$ realized by
$\mathcal F$. Then $P\in\mathcal A^L_\rho(x)$, $Q\in\mathcal A^R_\rho(x)$, and
$\widehat\rho=\operatorname{Intro}^{LR}_t(\rho,x,P,Q)$. We have
$\cost(\widehat{\mathcal F})=\cost(\mathcal F)+w(x)-2$, so $M_{LR}\le R_t[\widehat\rho]$.

In each case, the corresponding $M$-value is at most $R_t[\widehat\rho]$. Therefore
\[
\min\{M_0,M_L,M_R,M_{LR}\}\le R_t[\widehat\rho],
\]
completing the proof.
\end{proof}

\subparagraph{Forget Nodes.} We now proceed to the second type of nodes; for these, we also have a corresponding operator.

\begin{definition}[Forget Operator]
Let $t$ be a forget node, let $i$ be the forgotten color, and let $\rho$ be a record for bag $t-1$.
Let $x=\sel_\rho(i)$ and let $P\in\traces_\rho$ be the unique trace containing $x$.

Write $P=(y_1,\dots,y_r)$ and suppose $x=y_\ell$. Let $Q$ be the sequence
\begin{equation*}
Q=(y_1,\dots,y_{\ell-1},y_{\ell+1},\dots,y_r)
\end{equation*}
obtained from $P$ by deleting $x$. The sequence $Q$ may be empty; if it is nonempty, then it is an
active trace.

Then $\operatorname{Forget}_t(\rho)$ is the unique record $\widehat\rho$ satisfying
\begin{equation*}
\sel_{\widehat\rho}=\sel_\rho\restriction A_t.
\end{equation*}

If $Q=\emptyset$, then
\begin{equation*}
\traces_{\widehat\rho}=\traces_\rho\setminus\{P\},
\end{equation*}
\begin{equation*}
\forgotten_{\widehat\rho}(a,b)=
\begin{cases}
\min\{2,\forgotten_\rho(a,b)+1\},
& \begin{array}{@{}l@{}}
  \text{if } (a,b)=\\
  \quad(\startsLeft_\rho(P),\stopsRight_\rho(P)),
  \end{array}\\
\forgotten_\rho(a,b), & \text{otherwise,}
\end{cases}
\end{equation*}
and
\begin{equation*}
\data_{\widehat\rho}(R)=\data_\rho(R)
\quad\text{for every }R\in\traces_{\widehat\rho}.
\end{equation*}

If $Q\neq\emptyset$, then
\begin{equation*}
\traces_{\widehat\rho}=(\traces_\rho\setminus\{P\})\cup\{Q\},
\qquad
\forgotten_{\widehat\rho}=\forgotten_\rho,
\end{equation*}
\begin{equation*}
\data_{\widehat\rho}(R)=\data_\rho(R)
\quad\text{for every }R\in\traces_\rho\cap\traces_{\widehat\rho},
\end{equation*}
and if
\begin{equation*}
\data_\rho(P)=(a,b,c,d),
\end{equation*}
then
\begin{equation*}
\data_{\widehat\rho}(Q)
=
\left(
a,\,
b,\,
\begin{cases}
0, & \text{if } x=\first(P),\\
c, & \text{otherwise,}
\end{cases}
\,
\begin{cases}
0, & \text{if } x=\last(P),\\
d, & \text{otherwise.}
\end{cases}
\right).
\end{equation*}
\end{definition}

\begin{lemma}[Forget Recurrence]\label{lem:wts-forget}
Let $t$ be a forget node. Then for every record $\widehat\rho$ of bag $t$,
\begin{equation*}
R_t[\widehat\rho]
=
\min\Bigl\{
R_{t-1}[\rho]
\;\Bigm|\;
\widehat\rho=\operatorname{Forget}_t(\rho)
\Bigr\},
\end{equation*}
where the minimum ranges over all records $\rho$ of bag $t-1$.
\end{lemma}

\begin{proof}
Let $i$ be the forgotten color. Since $P_t=P_{t-1}$ at a forget node, forgetting $i$ changes only
which colors are active and does not change the vertex set of a partial solution.

First let $\mathcal F$ be a $(t-1)$-partial solution realizing a record $\rho$. The same forest is a
$t$-partial solution. The operator $\operatorname{Forget}_t$ deletes the selected vertex of color
$i$ from the active trace that contains it. If the trace becomes empty, then the whole component has
no active vertex and is counted in the appropriate forgotten class; otherwise the remaining active
subsequence is the new trace, and the two open-end bits are updated exactly according to whether the
deleted vertex was the first or last active vertex of the component. Hence $\mathcal F$ realizes
$\operatorname{Forget}_t(\rho)$ and the cost is unchanged. Therefore
\begin{equation*}
R_t[\widehat\rho]
\le
\min\{\,R_{t-1}[\rho]\mid \widehat\rho=\operatorname{Forget}_t(\rho)\,\}.
\end{equation*}

Conversely, let $\mathcal F$ be a $t$-partial solution realizing $\widehat\rho$ with cost
$R_t[\widehat\rho]$. Let $x$ be the unique vertex of color $i$ in $\mathcal F$. Consider the record
$\rho$ induced by the same forest $\mathcal F$ when the active set is $A_{t-1}=A_t\cup\{i\}$. Then
$\mathcal F$ realizes $\rho$, and by the definition of the forget operator we have
$\widehat\rho=\operatorname{Forget}_t(\rho)$. Therefore
$R_{t-1}[\rho]\le \cost(\mathcal F)=R_t[\widehat\rho]$, which proves the reverse inequality.
\end{proof}

\subparagraph{Wrapping Up.} It remains to state how to obtain our solution at the end of the dynamic programming run, and to combine these steps together in Theorem~\ref{thm:wts-fpt-pw-plus-max-universe-size}.
\begin{lemma}[Solution Extraction]\label{lem:wts-record-extraction}
The optimum value of the fixed \wtsLong instance equals
\begin{equation*}
\min\{\,R_q[\rho]\mid
\rho=(\emptyset,\emptyset,\emptyset,\delta)
\text{ and } \delta \text{ satisfies one of the following conditions}\,\},
\end{equation*}
where:
\begin{enumerate}
    \item $\delta(1,1)=1$ and
    \begin{equation*}
    \delta(1,0)=\delta(0,1)=\delta(0,0)=0;
    \end{equation*}

    \item
    \begin{equation*}
    \delta(1,0)+\delta(0,1)+\delta(1,1)\ge 2,
    \end{equation*}
    \begin{equation*}
    \delta(1,0)+\delta(1,1)\ge 1,
    \qquad
    \delta(0,1)+\delta(1,1)\ge 1.
    \end{equation*}
\end{enumerate}
\end{lemma}

\begin{proof}
At the final bag $X_q=\emptyset$ we have $A_q=\emptyset$ and $P_q=[d]$. Thus every $q$-partial
solution is a colorful directed linear forest containing exactly one annotated tuple of each color,
and all of its components are recorded only by the mapping $\delta$.

It remains to characterize when the forgotten components form a feasible family. If there is exactly
one component, then it can be both first and last if and only if it is of type $(1,1)$; this is
exactly Item~1. If the feasible order uses at least two components, then some component must be
usable as the first component, and some component distinct from it must be usable as the last
component. Since the counts are truncated at $2$, this is equivalent to Item~2: among the
components of types $(1,0)$, $(0,1)$, and $(1,1)$ there are at least two, at least one has first
coordinate $1$, and at least one has second coordinate $1$. Components of type $(0,0)$ can be
placed between the first and last components and impose no condition.

By \cref{lem:wts-forest}, minimizing $R_q[\rho]$ over exactly these final records gives the optimum
value of the \wtsLong instance.
\end{proof}

With this in hand, we can finally prove Theorem~\ref{thm:wts-fpt-pw-plus-max-universe-size}.

\thmwtspw*

\begin{proof}
Note that the given nice path decomposition can be assumed to have at most $2d+1$ bags.
We evaluate the record table along this nice path decomposition. At the first empty bag, we set
$R_1[\rho_\emptyset]=0$ and set all other records to $\infty$. The remaining bags are processed by
\cref{lem:wts-introduce,lem:wts-forget}. Correctness follows by induction over the bags using these
two recurrences. The optimum value is then obtained from \cref{lem:wts-record-extraction}.

It remains to bound the running time. Fix a bag $t$ and put $a\coloneqq |A_t|$. Since the width is
$p$, we have $a\le p+1$. For each active color $i\in A_t$, there are at most
$|T_i|\le |\mathcal U_i|^2\le u^2$ choices for $\sel(i)$, hence at most $u^{2a}$ choices for
$\sel$. Once $\sel$ is fixed, the active traces are obtained by ordering the $a$ selected annotated
tuples and choosing where traces end. This gives at most $a!\cdot 2^a$ possibilities. Each trace
has one of $16$ data values, and there are at most $a$ traces, giving at most $16^a$ choices. The
mapping $\forgotten:\{0,1\}^2\to\{0,1,2\}$ has $3^4=81$ choices. Therefore the number of records at
one bag is at most
\begin{equation*}
81\cdot u^{2a}\cdot a!\cdot 2^a\cdot 16^a
\le
81\cdot (p+1)!\cdot 2^{5p+5}\cdot u^{2p+2}.
\end{equation*}

At a forget node, every record of the previous bag produces one candidate successor record, and the
update takes $\mathcal O(p+1)$ time. At an introduce node, fix a previous record $\rho$ and the
introduced color $i$. There are at most $|T_i|\le u^2$ choices for the new annotated tuple $x$.
For fixed $x$, the sets $\mathcal A^L_\rho(x)$ and $\mathcal A^R_\rho(x)$ have size at most
$p+1$, so the four introduce cases yield at most $(p+2)^2$ candidate updates. Constructing each
candidate record takes $\mathcal O(p+1)$ time. Thus one previous record generates all its introduce
updates in $\mathcal O((p+2)^3u^2)$ time.

Multiplying the number of records per bag by the update time per record and by the at most $2d+1$
bags gives
\begin{equation*}
\mathcal O\bigl((p+2)^3\cdot (p+1)!\cdot 2^{5p}\cdot u^{2p+4}\cdot d\bigr). \qedhere
\end{equation*}
\end{proof}

\subparagraph{A Proof of \cref{lem:fpt-by-k-given-pathdecomp}.} We restate and prove the lemma below.

\lempwdp*

\begin{proof}
We first compute the boundary sets $\partial B_G(X)$ for all full subtrees $X$ of $T$ via
\cref{lem:compute-boundaries}. If the algorithm rejects, then the instance is negative.
Otherwise, this step takes time $\mathcal{O}(|E|+k|V_t|)$, and moreover
$|\partial B_G(X)|\le k+2$ for every full subtree $X$.

We now compute the compressed values $\copt_X$ bottom-up over the tree.
For a leaf $X=r$, all values $\copt_X(v_L,v_R)$ with
$v_L,v_R\in\partial B_G(X)\cup\{\lozenge\}$ can be computed directly from
\cref{def:opt_star,def:compressed-opt,lem:opt_equals_optstar} in time
$\mathcal{O}(|N_G(r)|+(k+3)^2)=\mathcal{O}(|N_G(r)|+k^2)$. Summed over all leaves, this costs
$\mathcal{O}(|E|+k^2|V_t|)$.

Let $X$ be a full subtree with children $X_1,\dots,X_c$. By
\cref{obs:intersection-equiv}, the given nice path decomposition of $\mathcal I_G(X)$ is also a nice path
decomposition of the interaction graph of the \wtsLong instance arising from
\cref{lem:compute-copt}. For each $i\in[c]$, the corresponding tuple set is $\mathcal U_{X_i}^2$,
where $\mathcal U_{X_i}=\partial B_G(X_i)\cup\{\lozenge\}$, and hence
$|\mathcal U_{X_i}|\le k+3$.

For every pair $v_L,v_R\in \partial B_G(X)\cup\{\lozenge\}$, we compute $\copt_X(v_L,v_R)$ by
solving the \wtsLong instance from \cref{lem:compute-copt}, that is, with
left set $\rho_X^{-1}(v_L)\cap \bigcup_{i=1}^c \mathcal U_{X_i}$ and right set
$\rho_X^{-1}(v_R)\cap \bigcup_{i=1}^c \mathcal U_{X_i}$.
By \cref{thm:wts-fpt-pw-plus-max-universe-size} with $d=c$, $p=12k+2$, and $u=k+3$, one such
computation takes time
\begin{equation*}
    c\cdot 2^{\mathcal O(k\log k)}.
\end{equation*}
Since there are at most $(k+3)^2$ pairs $(v_L,v_R)$, the total time spent on the node $X$ is still
\begin{equation*}
    c\cdot 2^{\mathcal O(k\log k)}.
\end{equation*}

Because every internal node of $T$ has at least two children, the number of internal nodes is at
most $|V_t|-1$. Hence
\begin{equation*}
    \sum_X c = |V(T)|-1 \le 2|V_t|-2.
\end{equation*}
Therefore, the total time required to compute all records is
\begin{equation*}
    \mathcal O\bigl(2^{\mathcal O(k\log k)}\cdot |V_t| + |E|\bigr).
\end{equation*}
Finally, since
$\partial B_G(T)=\emptyset$, we have
\begin{equation*}
    \copt_T(\lozenge,\lozenge)
    =
    \min_{v_L,v_R\in B_G(T)} \opt_T(v_L,v_R),
\end{equation*}
and by \cref{lem:opt_equals_optstar}, the instance is a positive instance if and only if
\begin{equation*}
    \copt_T(\lozenge,\lozenge)-|B_G(T)|\le k. \qedhere
\end{equation*}
\end{proof}

\section[A Single-Exponential Algorithm with Respect to the Maximum Degree Delta of T]{A Single-Exponential Algorithm with Respect to the Maximum Degree \boldmath $\Delta$ of $T$}
\label{sec:degree}

Towards our second result, we will once again leverage the \wtsLong perspective on \prob. For our initial considerations, let us once again fix an instance of \wtsLong, i.e., tuple sets $T_1,\dots,T_d$, weight functions $w_1,\dots,w_d$, a left set $\mathcal L$ and a right set $\mathcal R$. Let $\mathcal{U}$ denote the corresponding universe.

For any index set $I \subseteq [d]$, $x \in \mathcal U$ and an arbitrary permutation $\pi$ of $I$, we define
\[
 f_I(x) \coloneqq \texttt{\textup{WTS}}\bigl(T_{\pi(1)}, \dots, T_{\pi(|I|)}, w_{\pi(1)}, \dots, w_{\pi(|I|)}, \mathcal L, \{x\}\bigr),
\]
that is, the optimum value of the \wtsLong instance obtained
by restricting to the tuple sets and weight functions with indices in $I$, with left
set~$\mathcal L$ and right set~$\{x\}$. Note that $\pi$ carries no meaning beyond ensuring the tuple sets and weight functions are passed in a consistent order.

Our aim is to employ an alternative dynamic programming approach to solve \wts, tailored to the bounded-$\Delta$ setting. Towards formalizing our records, we introduce some technical notation that will be used solely in this section.
For a predicate $P$, we write $\langle P \rangle$ for 1 if $P$ is true and 0 if $P$ is false.
For a $d$-dimensional vector $\vec{v}$, set $I(\vec{v})$ records all indices $i$ where $\vec{v}_i = 1$, that is $I(\vec{v}) \coloneqq \{ i \in [d] \mid \vec{v}_i = 1 \}$.
Vector $\mathbf{e}_i$ denotes the $i$-th standard basis unit vector in $\mathbb{N}^d$, vector $\mathbf{0}$ is the $d$-dimensional zero vector.{} For a vector $\vec{v} \in \mathbb{N}^d$, $\lVert \vec{v} \rVert_1 $ denotes $\sum_{i=1}^d v_i$.

We define record table $R : \bigl( \{0,1\}^d \setminus \{\mathbf{0}\} \bigr) \times \mathcal{U} \to \mathbb{N} \cup \{\infty\}$ recursively as follows.
Let $\vec{v} \in \{0,1\}^d \setminus \{\mathbf{0}\}$ and $x \in \mathcal{U}$. Then, we set

\begin{equation*}
R[\vec{v}, x] =
\begin{cases}
\min\{\,w_i(t)\mid t\in T_i,\ \fst(t)\in\mathcal L,\ \snd(t)=x\,\}, & \vec{v}=\mathbf{e}_i,\\
\displaystyle
\min\limits_{\substack{i \in I(\vec{v}),\ t \in T_i,\ x' \in \mathcal{U}\\ \snd(t)=x}}
\bigl( w_i(t) + R[\vec{v}-\mathbf{e}_i, x'] - \langle\,x'=\fst(t)\neq \lozenge\,\rangle \bigr),
& \text{otherwise.}
\end{cases}
\end{equation*}
Intuitively, a record $R[\vec{v}, x]$ counts the minimum value of a solution that uses exactly the tuple sets indexed by $I(\vec{v})$ and ends with a tuple whose second entry is $x$, which we formalize in the next lemma.

\begin{lemma}[Semantics of Records]

\label{lem:wts-dp-correct}
    Let $\vec{v} \in \{0,1\}^d \setminus \{\mathbf{0}\}$ and let $I \coloneqq I(\vec{v})$.
    Then, for all $x \in \mathcal{U}$, we have
    $f_I(x) = R[\vec{v}, x]$.
\end{lemma}

\begin{proof}
    We perform induction over $s\coloneqq |I|$.
    In the base case, we have $s = 1$, i.e.,  $\vec{v} = \mathbf{e}_i$ for some $i \in [d]$, and thus $I = \{i\}$.
    Fix $x \in \mathcal{U}$. In the instance restricted to $I$, any feasible solution must
    consist of a single tuple $t \in T_i$ with $\fst(t)\in \mathcal L$ and $\snd(t)=x$.
    Since there is only one tuple, its value is $w_i(t)$. Hence
    \[
        f_I(x)=\min\{\,w_i(t)\mid t\in T_i,\ \fst(t)\in\mathcal L,\ \snd(t)=x\,\}.
    \]
    By the definition of the record table for $\vec{v}=\mathbf{e}_i$, we also have
    \[
        R[\mathbf{e}_i,x]=\min\{\,w_i(t)\mid t\in T_i,\ \fst(t)\in\mathcal L,\ \snd(t)=x\,\}.
    \]
    If there is no tuple $t\in T_i$ with $\fst(t)\in\mathcal L$ and $\snd(t)=x$, then both minima
    are $\infty$.
    Thus $f_I(x)=R[\vec{v},x]$.

    In the general case with $s \geq 2$, let $x \in \mathcal{U}$.
    To derive $R[\vec{v},x] = f_I(x)$, we show $R[\vec{v},x] \le f_I(x)$ and $R[\vec{v},x] \ge f_I(x)$ separately.

    \subparagraph*{\boldmath $R[\vec{v},x] \le f_I(x)$.}
    If $f_I(x)=\infty$, then the inequality $R[\vec{v},x] \le f_I(x)$ holds trivially. Hence, assume $f_I(x) < \infty$.
    In this case there exists an optimal solution for the instance on $I$ with right set $\{x\}$;
    fix such a solution and denote it by $t_{\sigma(1)}, \dots, t_{\sigma(s)}$, where
    $\sigma$ is a permutation of $I$.

\begin{align*}
    f_I(x)
    &= \val(t_{\sigma(1)}, \dots, t_{\sigma(s)}) \tag*{(1)}\\
    &= \val(t_{\sigma(1)}, \dots, t_{\sigma(s-1)}) + w_{\sigma(s)}(t_{\sigma(s)}) - \langle\,\snd(t_{\sigma(s-1)})=\fst(t_{\sigma(s)})\neq \lozenge\,\rangle \tag*{(2)}\\
    &\ge f_{I \setminus \{\sigma(s)\}}(\snd(t_{\sigma(s-1)})) \notag\\
    &\hspace{2em} +\, w_{\sigma(s)}(t_{\sigma(s)}) - \langle\,\snd(t_{\sigma(s-1)})=\fst(t_{\sigma(s)})\neq \lozenge\,\rangle \tag*{(3)}\\
    &= R\bigl[\vec{v}-\mathbf{e}_{\sigma(s)}, \snd(t_{\sigma(s-1)})\bigr] \notag\\
    &\hspace{2em} +\, w_{\sigma(s)}(t_{\sigma(s)}) - \langle\,\snd(t_{\sigma(s-1)})=\fst(t_{\sigma(s)})\neq \lozenge\,\rangle \tag*{(4)}\\
    &\ge \min_{\substack{i \in I,\ t \in T_i,\ x' \in \mathcal{U}\\ \snd(t)=x}}
        \bigl( R[\vec{v}-\mathbf{e}_i, x'] + w_i(t) - \langle\,x'=\fst(t)\neq \lozenge\,\rangle \bigr) \tag*{(5)}\\
    &= R[\vec{v}, x]. \tag*{(6)}
\end{align*}

Here, (1) uses that $t_{\sigma(1)},\dots,t_{\sigma(s)}$ is an optimal solution for the instance on the index set $I$ with right set $\{x\}$ (hence its value equals $f_I(x)$ and $\snd(t_{\sigma(s)})=x$);
(2) is the definition of $\val(\cdot)$ when appending the last tuple;
(3) uses that $\val(t_{\sigma(1)},\dots,t_{\sigma(s-1)})$ is at least the optimum $f_{I \setminus \{\sigma(s)\}}(\snd(t_{\sigma(s-1)}))$ for the subinstance with right set $\{\snd(t_{\sigma(s-1)}\}$;
(4) applies the induction hypothesis with index set $I \setminus \{\sigma(s)\}$ and right set $\{\snd(t_{\sigma(s-1)})\}$;
(5) holds because the displayed minimum ranges in particular over the triplet $(i,t,x')=(\sigma(s), t_{\sigma(s)}, \snd(t_{\sigma(s-1)}))$;
and (6) is the recursive definition of the record $R[\vec{v}, x]$.
Thus in all cases we have
$
R[\vec{v}, x] \le f_I(x).
$

\subparagraph*{\boldmath $R[\vec{v},x] \ge f_I(x)$.}
    By the recursive definition of the record,
\begin{align*}
R[\vec{v}, x]
&=
\min_{\substack{i \in I,\ t \in T_i,\ x' \in \mathcal{U}\\ \snd(t)=x}}
\Bigl( R[\vec{v}-\mathbf{e}_i, x'] + w_i(t) - \langle\,x'=\fst(t)\neq \lozenge\,\rangle \Bigr).
\tag*{(7)}
\end{align*}
If $R[\vec{v},x]=\infty$, then $f_I(x)\le R[\vec{v},x]$ holds trivially.
Hence we may assume $R[\vec{v},x]<\infty$. In particular, the minimum in (7) is then taken over a nonempty set and is finite. Choose $(i^\star,t^\star,x^\star)$ attaining the minimum.

By the induction hypothesis applied to the index set $I\setminus\{i^\star\}$ and right set $\{x^\star\}$,
\begin{equation*}
f_{I \setminus \{i^\star\}}(x^\star)
=
R[\vec{v}-\mathbf{e}_{i^\star}, x^\star] < \infty.
\tag*{(8)}
\end{equation*}
Let $\tau$ be a permutation of $I\setminus\{i^\star\}$
where $t_{\tau(1)}, \dots, t_{\tau(s-1)}$
is a solution for the index set $I\setminus\{i^\star\}$ and right set $\{x^\star\}$ with
\begin{equation*}
\val\bigl(t_{\tau(1)},\dots,t_{\tau(s-1)}\bigr)
=
f_{I \setminus \{i^\star\}}(x^\star).
\tag*{(9)}
\end{equation*}
Consider the concatenation
\[
\mathcal S \coloneqq \bigl(t_{\tau(1)},\dots,t_{\tau(s-1)},\, t^\star\bigr),
\]
which is feasible for the full instance on $I$ with right set $\{x\}$ (since $\snd(t^\star)=x$). Therefore,
\begin{align*}
f_I(x)
&\le \val(\mathcal S) \tag*{(10)}\\
&= \val\bigl(t_{\tau(1)},\dots,t_{\tau(s-1)}\bigr) + w_{i^\star}(t^\star)
   - \langle\,\snd(t_{\tau(s-1)})=\fst(t^\star)\neq \lozenge\,\rangle \tag*{(11)}\\
&= R[\vec{v}-\mathbf{e}_{i^\star}, x^\star]
   + w_{i^\star}(t^\star) - \langle\,x^\star=\fst(t^\star)\neq \lozenge\,\rangle \tag*{(12)}\\
&= R[\vec{v}, x]. \tag*{(13)}
\end{align*}

Here,
(10) uses that $f_I(x)$ is the minimum value of any solution with index set $I$ and right set $\{x\}$, and $\mathcal{S}$ is one such solution,
(11) is the definition of $\val(\cdot)$ when appending the last tuple;
(12)
applies (9) and (8) and
substitutes $\snd(t_{\tau(s-1)})=x^\star$.
Finally, (13) holds by the choice of $(i^\star,t^\star,x^\star)$ as a minimizer in (7).
Thus $f_I(x)\le R[\vec{v}, x]$.

Combining both directions yields $f_I(x)=R[\vec{v},x]$.
\end{proof}

Next, we show how to extract an optimal solution from the records.

\begin{lemma}[Solution Extraction]

\label{lem:wts-fpt-d-solution-extraction}
    The optimum value of the given \wtsLong instance is
    $
        \min_{x\in \mathcal R\cap\mathcal U} R[\mathbf{1}, x],
    $
    where $\mathbf{1}$ denotes the all-one vector.
\end{lemma}

\begin{proof}
    By definition, the optimum value of the given instance is
    \[
        \texttt{\textup{WTS}}(T_1,\dots,T_d,w_1,\dots,w_d,\mathcal L,\mathcal R)
        =
        \min_{x\in \mathcal R\cap\mathcal U} f_{[d]}(x).
    \]
    Applying \cref{lem:wts-dp-correct} with $\vec{v} = \mathbf{1}$ yields
    \[
        f_{[d]}(x) = R[\mathbf{1}, x]
        \qquad\text{for all } x\in\mathcal R\cap\mathcal U,
    \]
    which proves the claim.
\end{proof}

Computing all records thus yields:

\begin{theorem}

\label{thm:wts-fpt-d}
    \wtsLong can be solved in time $\mathcal{O}(d \cdot 2^d \cdot |\mathcal{U}|^2)$.
\end{theorem}

\begin{proof}
    Compute the table $R$ by iterating over all nonzero vectors
    $\vec v\in\{0,1\}^d$ in increasing order of their binary representation. This is safe, as
    setting a bit from one to zero always yields a smaller integer.

    For $\vec v=\mathbf e_i$, set $R[\mathbf e_i,x]=\infty$ for all $x\in\mathcal U$. Then, for
    every tuple $t\in T_i$ with $\fst(t)\in\mathcal L$, set
    \begin{equation*}
        R[\mathbf e_i,\snd(t)]
        \coloneqq
        \min\{R[\mathbf e_i,\snd(t)],w_i(t)\}.
    \end{equation*}
    This takes time $\mathcal O(|T_i|+|\mathcal U|)$.

    Now let $\lVert\vec v\rVert_1\ge 2$. Set $R[\vec v,x]=\infty$ for all $x\in\mathcal U$.
    For every $i\in I(\vec v)$, put $\vec v^{-i}\coloneqq \vec v-\mathbf e_i$ and compute
    \begin{equation*}
        \mu_i\coloneqq \min_{x'\in\mathcal U} R[\vec v^{-i},x']
    \end{equation*}
    in time $\mathcal O(|\mathcal U|)$. For every $a\in\mathcal U$, define
    \begin{equation*}
        \operatorname{best}_i(a)
        \coloneqq
        \begin{cases}
            \mu_i, & \text{if } a=\lozenge,\\
            \min\{\mu_i, R[\vec v^{-i},a]-1\}, & \text{if } a\neq\lozenge.
        \end{cases}
    \end{equation*}
    These values are exactly
    \begin{equation*}
        \min_{x'\in\mathcal U}
        \bigl(R[\vec v^{-i},x']-\langle\,x'=a\neq \lozenge\,\rangle\bigr).
    \end{equation*}
    Computing all values $\operatorname{best}_i(a)$ takes time $\mathcal O(|\mathcal U|)$.
    Then, for every tuple $t\in T_i$, set
    \begin{equation*}
        R[\vec v,\snd(t)]
        \coloneqq
        \min\{R[\vec v,\snd(t)],w_i(t)+\operatorname{best}_i(\fst(t))\}.
    \end{equation*}

    For fixed $\vec v$ and $i\in I(\vec v)$, the work is
    $\mathcal O(|\mathcal U|+|T_i|)\subseteq \mathcal O(|\mathcal U|^2)$, since
    $|T_i|\le |\mathcal U|^2$. Hence the total time over all pairs $(\vec v,i)$ with
    $i\in I(\vec v)$ is
    \begin{equation*}
        \mathcal O\left(
        |\mathcal U|^2
        \sum_{\vec v\in\{0,1\}^d} |I(\vec v)|
        \right)
        =
        \mathcal O(d\cdot 2^d\cdot |\mathcal U|^2).
    \end{equation*}
    By \cref{lem:wts-dp-correct}, this correctly computes the record table $R$.
    Finally, by \cref{lem:wts-fpt-d-solution-extraction}, the optimal value of the instance is
    $\min_{x\in\mathcal R\cap\mathcal U} R[\mathbf{1}, x]$.
\end{proof}

We can now apply this result to compute the $\copt$ recurrence given by \cref{lem:compute-copt} efficiently in a bottom-up manner, and thus obtain our second main result:

\thmdegree*

\begin{proof}
    Let $(G,T,k)$ be an instance of \prob. We compute the boundary sets $\partial B_G(X)$ for all
    full subtrees $X$ of $T$ via \cref{lem:compute-boundaries}. If the algorithm rejects, then the
    instance is a negative instance. Otherwise, this step takes time $\mathcal O(|E|+k|V_t|)$.

    We evaluate all compressed values $\copt_X$ bottom-up over $T$.
    For a leaf $X=r$, all values $\copt_X(v_L,v_R)$ with
    $v_L,v_R\in \partial B_G(X)\cup\{\lozenge\}$ can be computed directly from
    \cref{def:opt_star,def:compressed-opt,lem:opt_equals_optstar} in time
    \[
        \mathcal O\bigl(|N_G(r)|+(k+3)^2\bigr)
        =
        \mathcal O\bigl(|N_G(r)|+k^2\bigr).
    \]
    Over all leaves, this takes time
    \[
        \mathcal O\bigl(|E|+k^2|V_t|\bigr).
    \]

    Let $X$ be a full subtree with children $X_1,\dots,X_d$, where $d\le \Delta$.
    For each $i\in[d]$, set
    \[
        \mathcal U_{X_i}\coloneqq \partial B_G(X_i)\cup\{\lozenge\},
        \qquad
        \mathcal U_X\coloneqq \bigcup_{i=1}^d \mathcal U_{X_i}.
    \]
    Since $|\partial B_G(X_i)|\le k+2$ by \cref{lem:compute-boundaries}, we have
    $|\mathcal U_{X_i}|\le k+3$ for every $i$, and hence
    \[
        |\mathcal U_X|
        \le
        \sum_{i=1}^d |\partial B_G(X_i)|+1
        \le
        d(k+2)+1
        =
        \mathcal O(\Delta k).
    \]

    \begingroup
    \emergencystretch=2em
    Fix $v_L\in \partial B_G(X)\cup\{\lozenge\}$. We run the dynamic program from
    \cref{thm:wts-fpt-d} on the \wtsLong instance with tuple sets
    $\mathcal U_{X_1}^2,\dots,\mathcal U_{X_d}^2$, weight functions
    $\copt_{X_1},\dots,\copt_{X_d}$, left set
    $\rho_X^{-1}(v_L)\cap \mathcal U_X$, and universe $\mathcal U_X$.
    This computes, for every $x\in\mathcal U_X$, the optimum value $F_x$ for the corresponding
    singleton right set $\{x\}$ in time
    \par
    \endgroup
    \[
        \mathcal O\bigl(d\cdot 2^d\cdot |\mathcal U_X|^2\bigr)
        =
        \mathcal O\bigl(2^\Delta\cdot \Delta^3\cdot k^2\bigr).
    \]

    For each $v_R\in \partial B_G(X)\cup\{\lozenge\}$, \cref{lem:compute-copt} gives
    \[
        \copt_X(v_L,v_R)
        =
        \min_{x\in \rho_X^{-1}(v_R)\cap \mathcal U_X} F_x.
    \]
    Computing these minima for all choices of $v_R$ is dominated by the WTS computation.

    Since $|\partial B_G(X)\cup\{\lozenge\}|\le k+3$, we perform at most $k+3$ such WTS runs for
    the node $X$. Thus the total time spent on $X$ is
    \[
        \mathcal O\bigl(k\cdot 2^\Delta\cdot \Delta^3\cdot k^2\bigr)
        =
        \mathcal O\bigl(2^\Delta\cdot \Delta^3\cdot k^3\bigr).
    \]

    Since every internal node of $T$ has at least two children, we have
    $|V(T)|=\mathcal O(|V_t|)$. Consequently, the total time for evaluating all tables is
    \[
        \mathcal O\bigl(2^\Delta\cdot \Delta^3\cdot k^3\cdot |V_t|\bigr).
    \]
    Together with the boundary computation and the leaf computations, this yields
    \[
        \mathcal O\bigl(2^\Delta \cdot \Delta^3 \cdot k^3 \cdot |V_t| + |E|\bigr)
        \subseteq
        \mathcal O\bigl(2^\Delta \cdot \Delta^3 \cdot k^3 \cdot n + m\bigr),
    \]
    as claimed.

    Finally, since $\partial B_G(T)=\emptyset$, the instance is a positive instance if and only if
    \[
        \copt_T(\lozenge,\lozenge)-|B_G(T)|\le k. \qedhere
    \]
\end{proof}

We note that \cref{thm:thmdegree} cannot be improved to $2^{o(\Delta)}$ unless the \ETH fails (cf.~\cref{lem:eth-lower-probstar}).

\section{A Single-Exponential Algorithm for the Unconstrained Problem}

In this section, we show that \probStar{} admits a single-exponential fixed-parameter algorithm, which---recalling \cref{lem:eth-lower-probstar}---is tight under the \ETH.
To this end, we derive a linear kernel for \probStar{} and then solve it using our algorithm for \prob{} parameterized by the maximum degree of the tree.

We use the pathwidth-one splitting problem introduced by Baumann, Pfretzschner, and Rutter~\cite{povs} as a starting point: \povs{} asks, given a graph $G$, a set $S\subseteq V(G)$ and an integer $k$, whether $G$ can be transformed into a graph of pathwidth at most $1$ by at most $k$ vertex splits applied only to vertices of $S$ and the descendants produced from them by those splits. It is known (cf.\ \cref{prop:pw1-char} and~\cite{povs}) that
\probStar{} is equivalent to the restriction of \povs{} to bipartite graphs $G$ with bipartition $V_t,V_b$ and splittable set $S=V_b$.
Hence, throughout this section, we prove correctness of our reduction rules in the equivalent pathwidth-one formulation.
We remark that the linear kernel of Baumann, Pfretzschner, and Rutter~\cite{povs} (recently improved in~\cite{lwh-ikavstp-26}) is not directly suitable for our purposes, as some of Baumann, Pfretzschner, and Rutter's reduction rules do not preserve bipartiteness and the set $S$.
However, we can (and do) reuse those rules that preserve both properties.

\subparagraph{Preliminaries for the Kernel.}
Our notation follows that employed in the preceding work~\cite{povs}.
For $v\in V(G)$, let $d_G(v)\coloneqq |N_G(v)|$. A vertex is \emph{pendant} if it has degree $1$. For $A\subseteq V(G)$, let $\dist_G(v,A)$ denote the length of a shortest path from $v$ to a vertex of $A$, and let $\dist_G(v,A)=\infty$ if no such path exists. Define
$d_G^*(v)\coloneqq \bigl|\{u\in N_G(v)\mid d_G(u)>1\}\bigr|$,
$V_3(G)\coloneqq \{v\in V(G)\mid d_G^*(v)\ge 3\}$,
and
$\mu(G)\coloneqq \sum_{v\in V(G)} \max(d_G^*(v)-2,0)$.
An \emph{$N_2$-substructure} is a subgraph consisting of a vertex $r$, called its \emph{root}, adjacent to three vertices with degree at least two. Thus, a vertex $r$ is the root of an $N_2$-substructure if and only if $d_G^*(r)\ge 3$.
A connected graph is a \emph{caterpillar} if deleting all pendant vertices leaves a path. A connected graph is a \emph{pseudo-caterpillar} if deleting all pendant vertices leaves a cycle. The following lemma is adapted from~\cite{povs} and is in particular based on~\cite{emw-ecp-86}.

\begin{lemma}[Alternative Characterization]
\label{prop:pw1-char}
A bipartite graph $G=(V_t\cup V_b,E)$ admits a crossing-free two-layer drawing iff $\operatorname{pw}(G)\le 1$ iff $G$ is acyclic and contains no $N_2$-substructure.
\end{lemma}

\begin{figure}
    \centering
    \includegraphics[page=7]{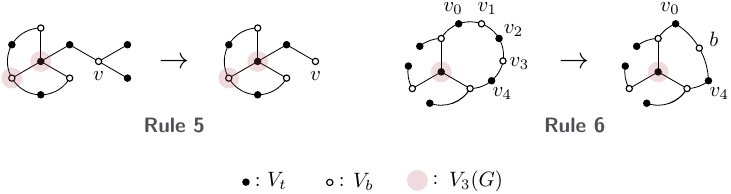}
    \caption{Example applications of Rules 5 and 6.}
    \label{fig:kernel_rules}
\end{figure}

\begin{definition}[Reduction rules for {\normalfont\probStar}]\label{def:classic-rules}
Let $(G=(V_t\cup V_b,E),k)$ be an instance of \probStar.
We use the following reduction rules. Rule~6 is applied only when Rules~1 and~5 are not applicable.
\begin{description}
    \item[Rule 1.] If $G$ contains a vertex with at least two pendant neighbors, delete all but one of them.
    \item[Rule 2.] If $G$ contains a connected component that is a caterpillar, delete it.
    \item[Rule 3.] If $G$ contains a connected component that is a pseudo-caterpillar, then if $k=0$, reject; otherwise delete it and decrease $k$ by $1$.
    \item[Rule 4.] If $\mu(G)>2k$, reject.
    \item[Rule 5.] If $G$ contains a vertex $v$ with $\dist_G(v,V_3(G))\ge 2$, delete all pendants adjacent to $v$.     \item[Rule 6.] If $G$ contains an induced path $v_0v_1v_2v_3v_4$ such that $v_0,v_2,v_4\in V_t$, $v_1,v_3\in V_b$, every $v_i$ has distance at least $2$ from $V_3(G)$, and $v_0,v_4$ are not pendant, delete $v_1,v_2,v_3$, add a new vertex $b\in V_b$, and add the edges $v_0b$ and $bv_4$.
\end{description}
\end{definition}

See~\cref{fig:kernel_rules} for an example of the application of Rules~5 and 6.

\begin{longlemma}\label{lem:old-rules-safe}
Rules~1--4 are safe.
\end{longlemma}

\begin{proof}
Rules~1 and~2 are precisely Reduction Rules~1 and~2 of~\cite{povs}, restricted to bipartite instances with splittable set $S=V_b$.

Rule~3 is Reduction Rule~3 of~\cite{povs}. In our setting, the rejection case of that rule never occurs because the spine of a pseudo-caterpillar in a bipartite graph is an even cycle and hence contains a vertex of $S=V_b$. Thus the rule always costs exactly one split; if $k=0$, the instance is negative, and otherwise deleting the component and decreasing $k$ by $1$ is safe.

Rule~4 is Reduction Rule~9 of~\cite{povs}.
\end{proof}

\begin{longlemma}\label{lem:rule5-safe}
Rule~5 is safe.
\end{longlemma}

\begin{proof}
Let $G^-$ be obtained from $G$ by deleting the pendant neighbors of a vertex $v$ with $\dist_G(v,V_3(G))\ge 2$.

For the forward direction, apply the same split sequence to $G^-$. The resulting graph is an induced subgraph of a graph of pathwidth at most $1$, and hence also has pathwidth at most $1$.

For the reverse direction, let $H^-$ be obtained from $G^-$ by at most $k$ allowed splits with $\operatorname{pw}(H^-)\le 1$. We reinsert the deleted pendant vertices and assign all their edges to one descendant $v^\star$ of $v$.

If some descendant of $v$ has degree at least $2$ in $H^-$, choose it as $v^\star$. Then only pendant neighbors are added to a vertex that already has degree at least $2$ in $H^-$. Thus no vertex of $H^-$ changes from pendant to non-pendant, and no vertex of $H^-$ gains a new non-pendant neighbor.

Otherwise, every descendant of $v$ has degree at most $1$ in $H^-$. If every descendant of $v$ is isolated in $H^-$, then reinserting the deleted vertices creates a star centered at one descendant of $v$. Otherwise choose $v^\star$ adjacent to a vertex $u^\star$, and let $u$ be the ancestor of $u^\star$ in $G$. Since $uv\in E(G)$ and $\dist_G(v,V_3(G))\ge 2$, we have $u\notin V_3(G)$. As $v$ is a non-pendant neighbor of $u$ in $G$, the vertex $u$ has at most one non-pendant neighbor different from $v$. Hence, after $v^\star$ becomes non-pendant, the vertex $u^\star$ has at most two non-pendant neighbors.

The construction cannot create a cycle, since it only adds pendant vertices. The only vertex of $H^-$ that can gain a new non-pendant neighbor is $u^\star$, and the argument above shows that $u^\star$ has at most two non-pendant neighbors afterwards. Hence no $N_2$-substructure is created. By \cref{prop:pw1-char}, the resulting graph has pathwidth at most $1$. The number of splits is unchanged.
\end{proof}

\begin{longlemma}\label{lem:rule6-safe}
Rule~6 is safe.
\end{longlemma}

\begin{proof}
Let $G^-$ be obtained from $G$ by applying Rule~6 to the path $v_0v_1v_2v_3v_4$, and let $b$ be the new vertex. Since Rule~5 is not applicable, no $v_i$ has a pendant neighbor. Since every $v_i$ has distance at least $2$ from $V_3(G)$, no $v_i$ belongs to $V_3(G)$. Each internal vertex already has its two path neighbors, and $v_0$ and $v_4$ are not pendant by assumption; as a third neighbor would also be non-pendant and hence put the vertex in $V_3(G)$, it follows that each $v_i$ has degree exactly $2$ in $G$.

First suppose that $(G,k)$ is a yes-instance, and let $H$ be obtained from $G$ by at most $k$ allowed splits with $\operatorname{pw}(H)\le 1$.

If neither $v_1$ nor $v_3$ is split in $H$, replace the path $v_0v_1v_2v_3v_4$ by the path $v_0bv_4$. Any cycle using $v_0bv_4$ would give a cycle in $H$ by replacing $v_0bv_4$ with $v_0v_1v_2v_3v_4$, and no other cycle can be new. Moreover, no $N_2$-substructure is created: every old vertex keeps at most the same non-pendant neighbors, while $v_0$ and $v_4$ each only replace one non-pendant path neighbor by the non-pendant vertex $b$, and $b$ has degree $2$. Thus the resulting graph has pathwidth at most $1$ by \cref{prop:pw1-char}.

If at least one of $v_1$ and $v_3$ is split in $H$, delete $v_2$ and all descendants of $v_1$ and $v_3$, and add two descendants $b_0,b_4$ of $b$ with edges $v_0b_0$ and $b_4v_4$. The number of splits does not increase: the deleted vertices include the descendant of $v_2$, at least two descendants of one of $v_1,v_3$, and at least one descendant of the other, while the reduced witness uses only the two descendants $b_0,b_4$ for the single new vertex $b$. Deleting vertices cannot create a cycle, and adding pendant vertices $b_0,b_4$ to $v_0,v_4$ cannot create a cycle. Moreover, every old vertex has at most the same non-pendant neighbors as before; the only new neighbors are $b_0$ and $b_4$, and both are pendant. The new vertices themselves have at most one non-pendant neighbor. Hence no $N_2$-substructure is created, and $(G^-,k)$ is a yes-instance by \cref{prop:pw1-char}.

Conversely, suppose that $(G^-,k)$ is a yes-instance, and let $H^-$ be obtained from $G^-$ by at most $k$ allowed splits with $\operatorname{pw}(H^-)\le 1$. Since $b$ has degree $2$ in $G^-$ and splits creating isolated vertices are excluded, we may assume that either $b$ is not split, or it is split once into two descendants, one adjacent to $v_0$ and one adjacent to $v_4$.

If $b$ is not split, replace the path $v_0bv_4$ by the path $v_0v_1v_2v_3v_4$. This is a subdivision of the path between $v_0$ and $v_4$, so it creates no cycle; it also does not increase the number of non-pendant neighbors of any old vertex, and the inserted vertices have at most two non-pendant neighbors. Thus no $N_2$-substructure is created. By \cref{prop:pw1-char}, the resulting graph has pathwidth at most $1$.

If $b$ is split into descendants $b_0$ and $b_4$ adjacent to $v_0$ and $v_4$, respectively, delete $b_0$ and $b_4$, split $v_1$ into two descendants $v_1^0$ and $v_1^2$, and add the edges
\begin{equation*}
    v_0v_1^0,\qquad v_1^2v_2,\qquad v_2v_3,\qquad v_3v_4.
\end{equation*}
This replaces the split of $b$ by one split of $v_1$. The replacement attaches a pendant vertex to $v_0$ and replaces the pendant neighbor $b_4$ of $v_4$ by the path $v_4v_3v_2v_1^2$, so it cannot create a cycle. It also creates no $N_2$-substructure: all old vertices other than $v_4$ gain no non-pendant neighbor, $v_1^0$ and $v_1^2$ are pendant, $v_2$ has one non-pendant neighbor, $v_3$ has two non-pendant neighbors, and $v_4$ has at most two non-pendant neighbors because $v_4$ has exactly one neighbor outside the path in $G$. Hence the resulting graph has pathwidth at most $1$ by \cref{prop:pw1-char}.
\end{proof}

\begin{lemma}\label{lem:kernel-vertex-bound}
Assume Rules~1--6 have been applied exhaustively without rejecting. Then $|V_t| \le 17k$ and $|V_b| \le 17k$. Consequently, $|V(G)| \le 34k$.
\end{lemma}

\begin{proof}
If $V_3=\emptyset$, then every vertex $v$ satisfies $d_G^*(v)\le 2$. Hence, in every connected component, deleting all pendant vertices leaves a graph of maximum degree at most two, that is, a path or a cycle. Therefore every connected component is either a caterpillar or a pseudo-caterpillar. Since Rules~2 and~3 have been applied exhaustively, it follows that $G$ is empty. In particular, $|V_t|=0$.

So assume from now on that $V_3\neq\emptyset$. Define
\begin{equation*}
V_3^1 \coloneqq \{v\in V(G)\mid \dist_G(v,V_3)=1\}
\quad\text{and}\quad
V_3^{\ge 2} \coloneqq \{v\in V(G)\mid \dist_G(v,V_3)\ge 2\},
\end{equation*}
where vertices in connected components disjoint from $V_3$ have distance $\infty$ from $V_3$.
Furthermore, let $P\subseteq V(G)$ be the set of pendant vertices.

Then
\begin{equation*}
V_3\cap V_t,\quad
V_3^1\cap P\cap V_t,\quad
(V_3^1\setminus P)\cap V_t,\quad
V_3^{\ge 2}\cap P\cap V_t,\quad
(V_3^{\ge 2}\setminus P)\cap V_t
\end{equation*}
forms a partition of $V_t$.

By Rule~4, we have $\mu(G)\le 2k$. Thus, by the inequalities proved at the beginning of \cite[Lemma~3]{povs},
\begin{equation}\label{eq:v3-bounds}
|V_3|\le 2k
\qquad\text{and}\qquad
\sum_{v\in V_3} d_G^*(v)\le 6k.
\end{equation}

We now bound the five parts of the partition.

\begin{description}
    \item[$V_3\cap V_t$ and $V_3^1\cap P\cap V_t$:]
    Every vertex in $V_3^1\cap P\cap V_t$ is a pendant neighbor of a vertex in $V_3\cap V_b$. By Rule~1, each vertex has at most one pendant neighbor. Hence
    \begin{equation*}
    |V_3\cap V_t| + |V_3^1\cap P\cap V_t|
    \le |V_3\cap V_t| + |V_3\cap V_b|
    = |V_3|
    \le 2k.
    \end{equation*}

    \item[$(V_3^1\setminus P)\cap V_t$:]
    Every vertex in $(V_3^1\setminus P)\cap V_t$ is adjacent to some vertex of $V_3\cap V_b$ and has degree greater than one. Therefore it is counted by $d_G^*(u)$ for at least one vertex $u\in V_3\cap V_b$. Thus
    \begin{equation*}
    |(V_3^1\setminus P)\cap V_t|
    \le \sum_{u\in V_3\cap V_b} d_G^*(u).
    \end{equation*}

    \item[$V_3^{\ge 2}\cap P\cap V_t$:]
    Let $x\in V_3^{\ge 2}\cap P\cap V_t$, and let $y$ be its unique neighbor.
    Since $x$ has distance at least two from $V_3$, we have $y\notin V_3$.
    Moreover, by Rule~5, no vertex of $V_3^{\ge 2}$ is adjacent to a pendant vertex. Since $x$ is pendant and adjacent to $y$, it follows that $y\notin V_3^{\ge 2}$.
    Hence $y\in V_3^1$.

    Also, $y\notin P$, since otherwise $x$ and $y$ would form a connected component on two vertices, which is a caterpillar and would have been deleted by Rule~2.
    Therefore $y\in (V_3^1\setminus P)\cap V_b$.

    By Rule~1, every vertex of $(V_3^1\setminus P)\cap V_b$ has at most one pendant neighbor. Hence
    \begin{equation*}
    |V_3^{\ge 2}\cap P\cap V_t|
    \le |(V_3^1\setminus P)\cap V_b|.
    \end{equation*}
    Every vertex in $(V_3^1\setminus P)\cap V_b$ is adjacent to some vertex of $V_3\cap V_t$ and has degree greater than one, so
    \begin{equation*}
    |(V_3^1\setminus P)\cap V_b|
    \le \sum_{u\in V_3\cap V_t} d_G^*(u).
    \end{equation*}
    Consequently,
    \begin{equation*}
    |V_3^{\ge 2}\cap P\cap V_t|
    \le \sum_{u\in V_3\cap V_t} d_G^*(u).
    \end{equation*}

    \item[$(V_3^{\ge 2}\setminus P)\cap V_t$:]
    Let
    \begin{equation*}
    H\coloneqq G[V_3^{\ge 2}\setminus P].
    \end{equation*}
    We first show that every vertex of $H$ has degree exactly two in $G$.

    Let $v\in V(H)$.
    If $d_G(v)=0$, then the connected component of $v$ is a single vertex, hence a caterpillar, contradicting the exhaustive application of Rule~2.
    If $d_G(v)=1$, then $v$ is pendant, contradicting $v\notin P$.
    If $d_G(v)\ge 3$, then Rule~5 implies that $v$ has no pendant neighbor, so $d_G(v)=d_G^*(v)\ge 3$, and therefore $v\in V_3$, contradicting $v\in V_3^{\ge 2}$.
    Hence $d_G(v)=2$ for every vertex $v\in V(H)$.

    It follows that every connected component of $H$ is a path or a cycle.
    A cycle component is impossible, as it would then also be a connected component of $G$, and as cycles are pseudo-caterpillars, this is impossible by Rule~3.
    Therefore every connected component of $H$ is a path.

    Now let $C$ be a connected component of $H$. Since every vertex of $C$ has degree two in $G$ and $C$ is a path, exactly two edges leave $C$, one from each endpoint of $C$.
    Their endpoints outside $C$ lie in $V_3^1\setminus P$: they cannot lie in $V_3$, since vertices of $C$ have distance at least two from $V_3$; they cannot lie in $V_3^{\ge 2}\setminus P$, since then they would belong to the same connected component of $H$; and they cannot be pendant by Rule~5.

    Moreover, every vertex of $V_3^1\setminus P$ is adjacent to vertices of at most one connected component of $H$. Indeed, let $x\in V_3^1\setminus P$. Since $x$ has distance one from $V_3$, it has a neighbor in $V_3$. If $x$ were adjacent to vertices from two distinct connected components of $H$, or to two vertices of one connected component of $H$, then $x$ would have at least three neighbors of degree greater than one, namely one neighbor in $V_3$ and two neighbors in $H$. Thus $d_G^*(x)\ge 3$, so $x\in V_3$, a contradiction.

    Hence distinct connected components of $H$ correspond to disjoint pairs of vertices of $V_3^1\setminus P$. Therefore the number of connected components of $H$ is at most
    \begin{equation*}
    \frac{|V_3^1\setminus P|}{2}\le 3k.
    \end{equation*}

    Finally, each connected component of $H$ has at most five vertices, for otherwise, Rule~6 applies, as every path on at least six vertices contains five consecutive vertices whose first, third, and fifth vertices lie in $V_t$.
    Therefore every connected component of $H$ has at most five vertices, and so each connected component contributes at most three vertices to $V_t$. Since $H$ has at most $3k$ connected components, it follows that
    \begin{equation*}
    |(V_3^{\ge 2}\setminus P)\cap V_t|
    \le 3\cdot 3k
    = 9k.
    \end{equation*}
\end{description}

Summing up the above bounds, we obtain
\begin{align*}
|V_t|
&= |V_3\cap V_t| + |V_3^1\cap P\cap V_t| + |(V_3^1\setminus P)\cap V_t| \\
&\qquad + |V_3^{\ge 2}\cap P\cap V_t| + |(V_3^{\ge 2}\setminus P)\cap V_t| \\
&\le 2k + \sum_{u\in V_3\cap V_b} d_G^*(u) + \sum_{u\in V_3\cap V_t} d_G^*(u) + 9k \\
&\le 2k + \sum_{u\in V_3} d_G^*(u) + 9k \\
&\le 2k + 6k + 9k \\
&= 17k.
\end{align*}

The same argument, with $V_t$ and $V_b$ exchanged, gives $|V_b|\le 17k$. Hence
\begin{equation*}
|V(G)| = |V_t| + |V_b| \le 34k. \qedhere
\end{equation*}
\end{proof}

\begin{theorem}

\label{thm:linear_kernel}
\probStar{} parameterized by $k$ admits a kernel $(G'=(V_t'\cup V_b',E'),k')$ with
    $|V_t'|\le 17k' \le 17k$
and
    $|V_b'|\le 17k' \le 17k$.
Moreover, the kernel can be computed in time $\mathcal O(m^2)$.
\end{theorem}

\begin{proof}
Apply Rules~1--6 exhaustively, always applying Rule~6 only when Rules~1 and~5 are not applicable. By \cref{lem:old-rules-safe,lem:rule5-safe,lem:rule6-safe}, all rules are safe.

We implement this as follows. After each non-rejecting rule application, recompute degrees, pendant vertices, $V_3(G)$, distances from $V_3(G)$, connected components, and $\mu(G)$ in time $\mathcal O(n+m)$. Rules~1--5 can then be tested in time $\mathcal O(n+m)$.

To test Rule~6, consider the graph
\begin{equation*}
    H\coloneqq G[V_3^{\ge 2}\setminus P],
\end{equation*}
where $P$ is the set of pendant vertices. If Rules~1 and~5 are not applicable, then every vertex of $H$ has degree exactly two in $G$. Hence every connected component of $H$ can be traversed in linear time. Rule~6 is applicable exactly when such a path component contains five consecutive vertices in the pattern $V_t,V_b,V_t,V_b,V_t$. Thus Rule~6 can also be tested in time $\mathcal O(n+m)$.

Each non-rejecting rule application decreases $|V(G)|$. Hence there are at most $|V(G)|$ such applications, and the exhaustive reduction takes time $\mathcal O((n+m)^2)$.

Let $(G'=(V_t'\cup V_b',E'),k')$ be the resulting instance. By \cref{lem:kernel-vertex-bound}, if the instance is not rejected, then
\begin{equation*}
    |V_t'|\le 17k'
    \qquad\text{and}\qquad
    |V_b'|\le 17k'.
\end{equation*}
Since $k'\le k$, the claimed bounds follow.
\end{proof}

Finally, applying the $\mathcal{O}^*(2^\Delta)$ algorithm of \cref{thm:thmdegree} to this (linear) kernel yields:

\thmclassic*

\begin{proof}
First compute the kernel from \cref{thm:linear_kernel} in time $\mathcal O((n+m)^2)$. Let $(G'=(V_t'\cup V_b',E'),k')$ be the resulting instance. We have $|V_t'|\le 17k'\le 17k$.

If $|V_t'|\le 1$, the instance is trivial. Otherwise, let $T$ be the star with leaf set $V_t'$. Then \probStar{} on $(G',k')$ is exactly \prob{} on $(G',T,k')$, where the maximum degree of $T$ is $|V_t'|$. Since $G'$ is bipartite and $|V_t'|,|V_b'|\le 17k$, we have $|V(G')|\le 34k$ and $|E'|\le |V_t'||V_b'|\le 289k^2$. Applying \cref{thm:thmdegree} gives running time
\begin{equation*}
    \mathcal O\bigl(2^{|V_t'|}\cdot |V_t'|^3\cdot (k')^3\cdot |V(G')| + |E'|\bigr)
    =
    \mathcal O\bigl(2^{17k}\cdot k^7\bigr).
\end{equation*}
Since $|V_t'|\le 17k$, $k'\le k$, and $|V(G')|\le 34k$, the product $|V_t'|^3(k')^3|V(G')|$ is $\mathcal O(k^7)$.
Thus, in total, we can solve the instance in time $\mathcal O(2^{17k}\cdot k^7+(n+m)^2)$.
\end{proof}

Moreover, \cref{thm:thmclassic} cannot be improved to $2^{o(k)}$ unless the \ETH fails (cf.~\cref{lem:eth-lower-probstar}).

\section{The Implementation}
\label{sec:implementation}

We implemented the $\mathcal{O}^*(2^\Delta)$ algorithm of
\cref{thm:thmdegree} in \texttt{C++}. The source code, instances, and result files are available in the
\href{https://doi.org/10.17605/OSF.IO/2VCY9}{\textcolor[RGB]{45,85,130}{OSF supplement}}~\cite{supplement}.
All experiments were run on an Apple M1 Pro processor.
We first tested the implementation on the Human Reference Atlas ASCT+B tables
discussed in the introduction. We obtained \num{36} real-world instances with up
to \num{942} anatomical vertices, \num{188} cell types, and \num{2301} edges,
all solved in at most \SI{220.83}{\milli\second} on a single core. One resulting drawing for a
medium-sized instance is shown in \cref{fig:large}. We also tested the
implementation on synthetic instances, some orders of magnitude larger than the
real-world ones, using all cores (see \cref{fig:benchmark}).

\begin{center}
    \centering

  \includegraphics[page=2,width=.965\linewidth]{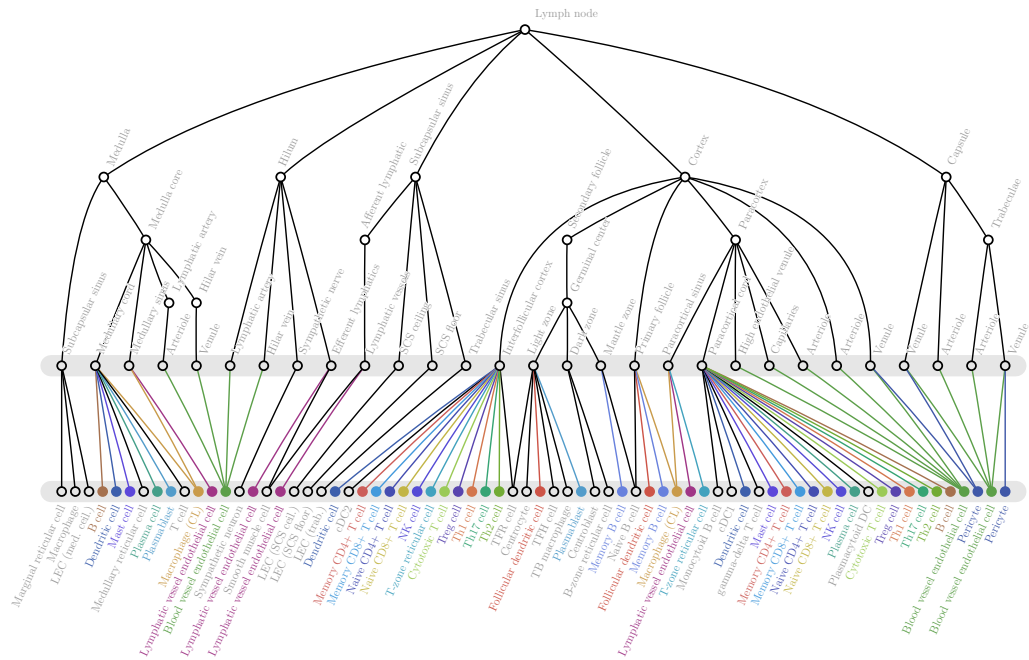}
    \captionsetup{hypcap=false}
    \captionof{figure}{An optimal solution for the lymph-node instance ($\Delta=6$, \num{26} splits), found in \SI{5.23}{\milli\second}.}
    \label{fig:large}
\end{center}

\begin{center}
    \centering
    \includegraphics[]{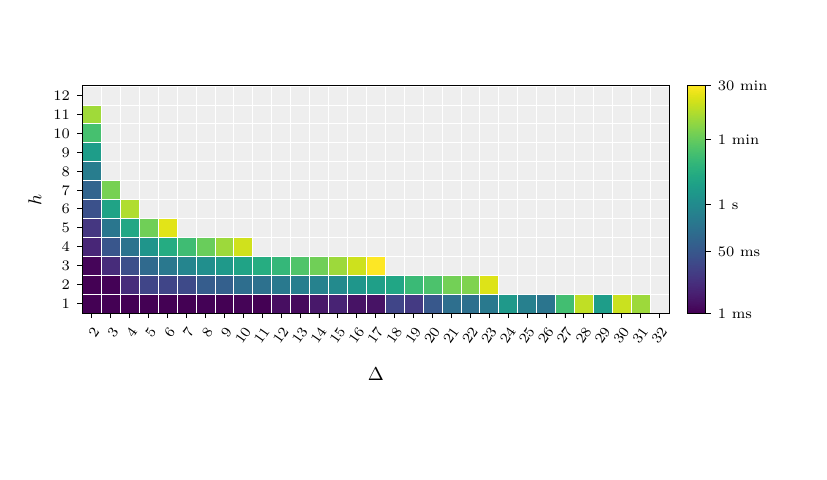}
    \captionsetup{hypcap=false}
\captionof{figure}{Benchmark on random instances with a full $\Delta$-ary tree $T$ of height
$h$ and $|V_t|=|V_b|$, where each bottom degree is sampled independently from a
geometric distribution with mean $2$. For each cell, one instance was sampled
and its running time is indicated by the color. Gray cells represent
computationally  infeasible runs, i.e., runs that did not finish within a
30-minute timeout or exhausted the available memory. For example, the cell
$h=4$, $\Delta=10$ corresponds to an instance with
$|V_t|=|V_b|=\Delta^h=10^4$ and $|E|=\num{19993}$, and was solved
in $10.33$ minutes with  split number~$\num{9724}$.
}

    \label{fig:benchmark}
\end{center}
\section{Concluding Remarks}
\looseness=-1
Our results delimit the state of the art for \problong, and moreover improve upon the previous algorithms~\cite{povs,lxhw-pap2dvsgs-24} for \probStar (which can be seen as a special case of \prob) by developing a single-exponential algorithm. For \cref{thm:thmsplits}, it would be interesting for future work to understand whether the running time could be improved to single-exponential, matching both the ETH lower bound from Lemma~\ref{lem:eth-lower-probstar} and the running time of our other two results. On the other hand, the running time for Theorem~\ref{thm:thmdegree} seems to be essentially optimal: it cannot be improved to subexponential under the Exponential Time Hypothesis, and moreover it seems challenging to reduce the base of the exponent below $2$ via the currently employed approaches. Indeed, it is not difficult to observe that an improvement to our \wts algorithm from $\mathcal{O}^*(2^d)$ to $\mathcal{O}^*(1.99^d)$ (where $d$ is the number of tuple sets) would yield an $O^*(1.99^n)$ algorithm for \textsc{Directed Hamiltonian Path}, on which a recent article~\cite{DBLP:journals/jcss/KowalikLNSW22} commented:

\begin{quote}
\emph{[T]he question about its existence is a major open problem which in the last 58 years has seen some progress
only for special graph classes, like bipartite graphs.}
\end{quote}

\clearpage
\newpage

\bibliography{refs}

\end{document}